\documentclass[12pt]{article}
\usepackage{amssymb}
\usepackage{latexsym}
\usepackage{amsmath}
\usepackage{amsfonts}
\usepackage{setspace}
\usepackage{enumerate}
\usepackage{verbatim}
\usepackage{graphicx}
\usepackage[multiple]{footmisc}
\usepackage{ifpdf}%
\providecommand{\U}[1]{\protect\rule{.1in}{.1in}}
\newtheorem{theorem}{Theorem}

\newtheorem{corollary}{Corollary}

\newtheorem{definition}{Definition}

\newtheorem{lemma}{Lemma}

\newtheorem{proposition}{Proposition}
\newtheorem{remark}{Remark}

\newcommand{\XEFF}{\textsl{ex-post efficiency}}
\newcommand{\XNW}{\textsl{ex-post non-wastefulness}}
\newcommand{\XWNW}{\textsl{ex-post weak non-wastefulness}}

\newcommand{\ETE}{\textsl{equal treatment of equals}}
\newcommand{\SETE}{\textsl{equal treatment of almost equals}}
\newcommand{\SP}{\textsl{strategy-proofness}}

\begin{document}

\title{\textbf{Strategy-Proof and Minimally Wasteful Random Assignment}\thanks{Support by the Deutsche Forschungsgemeinschaft through CRC TRR 190 (project number 280092119), the SSHRC (Canada) under Insight Grant 435-2023-0129, and the FRQ (Qu\'ebec) under Soutien aux \'equipes de recherche / Universitaire - nouvelle équipe 367853 is gratefully acknowledged.}}
\author{Christian Basteck\thanks{WZB Berlin Social Science Center, Berlin, Germany. Email: christian.basteck@wzb.eu.} \and
Lars Ehlers\thanks{D\'epartement de Sciences \'Economiques and CIREQ, Universit\'{e} de Montr\'{e}al, Montr\'{e}al, QC H3C 3J7, Canada. Email: lars.ehlers@umontreal.ca.}}
\maketitle

\begin{abstract}
\noindent We study random assignment of indivisible objects among a set of agents with strict preferences and outside options. When agents may rank some objects as unacceptable, we consider different notions of measuring waste of object(s) from an ex-ante perspective. The most natural %na\"{\i}ve
one is $q$-agent-object-waste whereby both one agent and one of his acceptable objects are unassigned with at least probability $q$. 
%Then it is feasible to transfer the share $q$ of the object to the agent.
On the one hand, we show that any mechanism satisfying \SETE\ (whereby any two agents with identical rankings over objects receive the same probability shares for objects they both regard acceptable), \SP\ and \XWNW\ (whereby in any assignment in the support we cannot have that both one agent and one of his acceptable objects are unassigned) must be $q$-agent-object-wasteful with $q\geq \frac{1}{6}$. On the other hand, we show that random serial dictatorship (RSD) attains the minimal bound of $\frac{1}{6}$ in this class for four agents or three objects. % These results remain unchanged for \ETE, \WEF, \SP\ and \XWNW.
In addition, we
consider $q$-object-wastefulness where an object remains unassigned with probability $q$,
while agents, who consider it acceptable, remain unassigned with aggregate probability
$q$.
We again show that RSD attains the minimal bound of $\frac{1}{4}$ in this class of mechanisms with respect to $q$-object-wastefulness for three objects. 
%Finally, we show that random deferred acceptance (RDA) is at-most-$q$-agent-object-wasteful with $q<\frac{1}{6}$ but at the same time $q$-object-wasteful with $q\geq \frac{1}{4}$.\smallskip
Finally, we show that random deferred acceptance (RDA) may be strictly less agent-object-wasteful than RSD (but at the cost of violating equal treatment of almost equals).\smallskip

\smallskip

\end{abstract}

\section{Introduction}

Over the past 30 years the assignment of indivisible objects has received considerable attention in economics. Applications range from allocating public school seats or child care places, assigning public houses among socially disadvantaged groups, matching organs among patients to the distribution of vaccines. All these applications have in common that monetary
transfers are fixed or absent and that prices cannot be used in the design of those markets.
In the pursuit of socially desirable objectives, the market designer has to devise a
mechanism from (reported) ordinal preferences to chosen assignments.
Oftentimes agents may
rank objects unacceptable as for instance in school choice a private school may be preferred
to a public school while in kidney exchange a compatible patient-donor pair may rank certain
donors unacceptable (and the market design ensures their participation by assigning them
only acceptable objects). %Due to the nature of indivisibilities,
There are two opposing designs of those markets: deterministic versus random.

For deterministic assignment three approaches are well-established and used in applications: serial dictatorship (SD), top-trading-cycles (TTC) and deferred-acceptance (DA). Under SD, agents choose the objects in a prespecified order. In applications, this order is determined via objective criteria such as wait lists or when the request was received. Under TTC, endowments are constructed by assigning the objects to agents and then the TTC-mechanism by Shapley and Scarf (1974) is executed. A general approach encompassing SD and TTC are hierarchical exchange rules (HE) by Papai (2000). Under DA, priorities for objects are constructed and then the DA-mechanism is executed. The (dis)advantages of these mechanisms are as follows. On the one hand, HE are \textsl{ex-post efficient} and unstable whereas DA violates \XEFF\ and is stable. On the other hand all those mechanisms are non-wasteful, which means that no object is both unassigned and preferred by an agent to his assignment, i.e., any object has zero waste. For deterministic mechanisms ex-ante and ex-post versions coincide with respect to efficiency and (weak) non-wastefulness. To the best of our knowledge, wasteful mechanisms are not used in real life for the design of deterministic assignment.

The main shortcoming of deterministic mechanisms is that they are extremely unfair. For instance, if all agents have the same most preferred object, then any fair deterministic assignment would have to waste this object in the absence of monetary transfers. This is the reason for implementing random assignment mechanisms in those contexts. % even though politicians often shy away and favor deterministic assignment mechanisms.
In some sense, randomness allows to substitute for absence of (continuous) prices. Na\"{i}vely mixing the above deterministic mechanisms results in random serial dictatorship (RSD), random top-trading-cycles (RTTC), more generally in random hierarchical exchange (RHE)\footnote{Bade (2020) has shown by using Pycia and \"Unver (2017) that RHE always results in RSD.} and random deferred acceptance (RDA). While those mechanisms satisfy \ETE\ and \SP, unfortunately, by considering random mechanisms we necessarily introduce waste of objects, as Martini (2016) has shown that any mechanism satisfying \ETE\ and \SP\ must be wasteful in a strong sense whereby an agent prefers an object, which is unassigned with positive probability, to one of the objects to which he is assigned with positive probability.

Throughout we consider the weaker (unambiguous) property by considering an object wasted if it is unassigned with positive probability and some agent, who is unassigned with positive probability, desires this object. This notion is easy to check as it only requires to do so for any agent-object pair where both are unassigned with positive probability.
Note that waste in this weak sense implies that the random assignment fails to maximize
the expected size of an assignment.
%Then the random assignment does not ex-ante assign the maximum size of agents to their acceptable objects.

At the same time we consider \SETE, which means that whenever two agents rank all objects identically and their sets of acceptable objects differ by at most one object, then they shall receive the same probability shares for objects they both regard acceptable.
We show that this does not allow us to escape the impossibility, i.e., there is no mechanism satisfying \SETE, \SP\ and ex-ante weak non-wastefulness, which is somewhat unsurprising given that contrary to the deterministic approach that relinquishes \ETE, the literature on random assignment established numerous impossibilities regarding the existence of mechanisms satisfying non-manipulability (\SP), \ETE\ and different efficiency notions.

Our main purpose is to show the existence of minimally wasteful mechanisms in various classes of mechanisms. This requires to find both the minimal bound on waste in a certain class of mechanisms and to show the existence of a mechanism attaining this bound.
We introduce different notions of waste regarding object(s), agent(s) and across them. Given $q\in [0,1]$, $q$-agent-object-wastefulness means that at some preference profile both an agent and one of his acceptable objects are unassigned with probability greater than or equal to $q$. On the one hand, we show that any mechanism satisfying (i) \SETE, \SP\ and \XWNW\ or (ii) \ETE, \SP\ and \XEFF\ is $q$-agent-object-wasteful with $q\geq \frac{1}{6}$. On the other hand, for three objects or four agents we show that RSD is at-most-$\frac{1}{6}$-agent-object-wasteful. Hence, RSD obtains the minimal bound regarding $q$-agent-object-wastefulness in the class of mechanisms satisfying \SETE, \SP\ and \XEFF.
%, and respectively, in the one satisfying \ETE, \WEF, \SP\ and \XWNW.

We also consider the notion of $q$-object-wastefulness whereby the share $q$ of the unassigned part of the object can be feasibly transferred to the agents preferring it over being unassigned (while not changing the random assignment of any other object).
We show $\frac{1}{4}$-object-wastefulness of any mechanism in the class satisfying (i) \SETE, \SP\ and \XWNW, %(ii) \ETE, \WEF, \SP\ and \XWNW,
or (ii) \ETE, \SP\ and \XEFF.
For three objects we obtain that RSD is at-most-$\frac{1}{4}$-object-wasteful, i.e., this bound is the minimal one for $q$-object-wastefulness in these two classes of mechanisms.
We finally consider random deferred acceptance (RDA), which violates both \SETE\ and \XEFF, and show for four agents and three objects that RDA is at-most-$q$-agent-object-wasteful with $q<\frac{1}{6}$ while RDA is $q$-object-wasteful with $q\geq \frac{1}{4}$.

%We further establish ex-ante weak non-wastefulness of any such mechanism adding to Martini (2016). We also consider stochastic dominance among %strategy-proof mechanisms and show that if one mechanism stochastically dominates another one, then there exists a profile whereby the first %mechanism object-by-object dominates the second one meaning that for any acceptable object of an agent she is assigned to that object with weakly %greater probability under the first mechanism in comparison to the second mechanism.

The paper is organized as follows. Section \ref{sec2} introduces random assignments, their properties and several prominent mechanisms.
%Section \ref{sec3} contains our taxonomy of measuring waste of object(s).
Section \ref{sec4} establishes that $\frac{1}{6}$ is the minimal bound for $q$-agent-object-wastefulness in the two classes of mechanisms satisfying (i) \SETE, \SP\ and \XWNW\ and (ii) \ETE, \SP\ and \XEFF, and at-most-$\frac{1}{6}$-agent-object-wastefulness of RSD for four agents or three objects.
Section \ref{sec5} establishes that $\frac{1}{4}$ is the minimal bound for $q$-object-wastefulness in the two classes of mechanisms, and
at-most-$\frac{1}{4}$-object-wastefulness of RSD for three objects.
Section \ref{secRDA} relates our main results to the random deferred-acceptance mechanism.
Section \ref{sec6} discusses our results with respect to ex-ante (weak) non-wastefulness.
Section \ref{sec7} concludes. The Appendix contains the proofs omitted from the main text.

\section{Model}\label{sec2}

Let $N=\{1,\ldots,n\}$ denote the set of agents and $O=\{o_1,\ldots,o_m\}$ denote the finite set of objects. %For the moment, suppose that all capacities are one, i.e., $q_o=1$ for all $o\in O$,
%Throughout the main text we suppose $|N|=|O|\geq 3$ and allow for unequal numbers of agents and objects in the Appendix.
Each agent $i$ has strict preferences over $O\cup\{i\}$ where $i$ stands for being unassigned; let $R_i$ denote the corresponding linear order\footnote{Thus, $R_i$ is (i) complete, (ii) transitive and (iii) antisymmetric ($xR_iy$ and $yR_ix$ implies $x=y$).} and write $P_i$ for its asymmetric part (where $xP_iy$ is defined by $xR_iy$ and $x\neq y$). Let $\mathcal{R}^i$ denote the set of all strict preferences of agent $i$ over $O\cup \{i\}$, i.e., where some objects may be unacceptable.
%Let $\mathcal{W}^i$ denote the set of all strict preferences of agent $i$ over $O\cup \{i\}$.
Let $\mathcal{R}^N=\times_{i\in N}\mathcal{R}^i$ denote the set of all preference profiles $R=(R_1,\ldots,R_n)$, which we call the full domain.
%We denote this domain by $\mathcal{R}^N=\times_{i\in N}\mathcal{R}^i$ and refer to it as the acceptable domain, as no agent would ever dispose of any assigned object. We call $\mathcal{W}^N$ the full domain.
%A strict priority ranking over $N$ is denoted by $\succ$. Let $\mathcal{L}$ denote the set of all strict priority rankings.

An assignment is a mapping $\mu:N\rightarrow O\cup N$ such that\footnote{We will use throughout the convention to write $\mu_i$ instead of $\mu(i)$ for any $i\in N$.} $\mu_i\in O\cup\{i\}$ for all $i\in N$ and $\mu_i\neq \mu_j$ for all $i\neq j$.
Let $\mathcal{M}$ denote the set of all assignments.

An assignment $\mu$ is efficient under $R$ if there exists no $\mu'\in \mathcal{M}$ such that $\mu_i'R_i\mu_i$ for all $i\in N$ and $\mu_j'P_j\mu_j$ for some $j\in N$. Let $\mathcal{PO}(R)$ denote the set of all efficient assignments under $R$.
%An assignment $\mu$ is weakly efficient under $R$ if there exists no $\mu'\in \mathcal{M}$ such that $\mu_i'P_i\mu_i$ for all $i\in N$. Let %$\mathcal{WPO}(R)$ denote the set of all weakly efficient assignments under $R$.

An assignment $\mu$ is non-wasteful under $R$ if for all $i\in N$ and all $x\in O\cup \{i\}$, $xR_i\mu_i$ implies there exists $j\in N$ with $\mu_j=x$. Note that this implies $\mu_iR_ii$. Let $\mathcal{NW}(R)$ denote the set of all non-wasteful assignments under $R$.

An assignment $\mu$ is weakly non-wasteful under $R$ if for all $i\in N$ and all $x\in O\cup \{i\}$,  $xR_i\mu_i$ and $i R_i \mu_i$ together imply that there exists $j\in N$ with $\mu_j=x$. Again this implies $\mu_iR_ii$. Here we consider an object to be wasted if it is unassigned but desired by an \emph{unassigned} agent or if an agent is assigned an unacceptable object.
Let $\mathcal{WNW}(R)$ denote the set of all weakly non-wasteful assignments under $R$. Note that verifying weak non-wastefulness only requires the knowledge of each agent's acceptable objects but no knowledge of the ranking among them (which is necessary to determine non-wastefulness or efficiency of a deterministic allocation). Finally, let $\mathcal{IR}(R)$ denote the set of all individually rational assignments, i.e., $\mu\in \mathcal{IR}(R)$ if and only if $\mu_iR_ii$ for all $i\in N$.

For any profile $R$, we have $\mathcal{PO}(R)\subseteq \mathcal{NW}(R)\subseteq \mathcal{WNW}(R)\subseteq \mathcal{IR}(R)$.
%, and there is no relation between non-wastefulness and weak efficiency.

Let $\Delta(\mathcal{M})$ denote the set of all probability distributions over $\mathcal{M}$. Given $p\in \Delta(\mathcal{M})$, let $p_{ia}$ denote the associated probability of $i$ being assigned $a$ and refer to $p_i=(p_{ia})_{a\in O\cup \{i\}}$ as agent $i$'s (individual) random assignment. Let $supp(p)$ denote the support of $p$. Then (i) $p$ is ex-post efficient under $R$ if $supp(p)\subseteq \mathcal{PO}(R)$,
%(ii) $p$ is ex-post weakly efficient under $R$ if $supp(p)\subseteq \mathcal{WPO}(R)$,
(ii) $p$ is ex-post non-wasteful under $R$ if $supp(p)\subseteq \mathcal{NW}(R)$, (iii) $p$ is ex-post weakly non-wasteful under $R$ if $supp(p)\subseteq \mathcal{WNW}(R)$, and (iv) $p$ is individually rational under $R$ if $supp(p)\subseteq \mathcal{IR}(R)$. Now if $p$ is individually rational, then for all $i\in N$ and $x\in O$ such that $iP_ix$ we have $p_{ix}=0$.

For all $i\in N$, all $R_i\in \mathcal{R}^i$ and all $x\in O\cup \{i\}$, let $B(x,R_i)=\{y\in O\cup \{i\}:yR_ix\}$.
Then given any $p,q\in \Delta(\mathcal{M})$, $p_i$ stochastically $R_i$-dominates $q_i$ if for all $x\in O\cup \{i\}$,
\[ \sum_{y\in B(x,R_i)} p_{iy} \geq \sum_{y\in B(x,R_i)} q_{iy}. \]
A random assignment $p$ stochastically $R$-dominates (or sd-dominates) another random assignment $q$ if $p_i$ $R_i$-dominates $q_i$ for all $i\in N$; $p$ strictly stochastically $R$-dominates $q$ if moreover $p_i\neq q_i$ for some $i\in N$. A random assignment $p$ is stochastic dominance (sd)-efficient if there is no random assignment $q$ that strictly stochastically $R$-dominates it.\footnote{Bogomolnaia and Moulin (2001) refer to this as ``ordinal efficiency''. It implies Pareto-efficiency with respect to expected utilities for some von Neumann-Morgenstern-representations of agents' ordinal preferences over objects.} Given two random assignments $p$ and $q$, we say that $p$ and $q$ are \textsl{welfare-equivalent} if $p_i=q_i$ for all $i\in N$.\footnote{Some papers directly define a bistochastic matrix $(p_{ia})_{i\in N, a\in O}$ rather than a random assignment per se, i.e., a convex combination of deterministic assignments. Nonetheless, corresponding random assignments exist as any bistochastic matrix $(p_{ia})_{i\in N, a\in O}$ can be decomposed as a convex combination of deterministic assignments by the Birkhoff-von Neumann Theorem (Birkhoff, 1946). Abdulkadiro\u{g}lu and S\"onmez (2003) observe that an ex-post efficient random assignment may be welfare-equivalent to a random assignment with support contained in the set of inefficient assignments so that $p_i=q_i$ for all $i\in N$ does not imply $p=q$. Furthermore, Zhang (2023) considers different notions of non-wastefulness, with the weakest one being ex-post non-wastefulness: a random assignment is ex-post-non-wasteful if it has a decomposition of non-wasteful assignments.} %The latter also observed by \cite{pyciatroyan2023} for RSD.}%This is an ex-ante efficiency notion before realizing the deterministic assignment to be implemented.

A mechanism (or rule) is a mapping $f:\mathcal{R}^N\rightarrow \Delta(\mathcal{M})$. Then $f(R)$ denotes the random assignment chosen for $R$, and $f_{ia}(R)$ denotes the probability of agent $i$ being assigned object $a$ whereas $f_{ii}(R)$ denotes the probability of agent $i$ being unassigned. For $i\in N$, $f_i(R)$ denotes the tuple of assignment probabilities $\left(f_{ia}(R)\right)_{a\in O\cup \{i\}}$, and for $a\in O$, $f_a(R)$ is defined accordingly as the tuple of probability shares with which $a$ is assigned to the various agents. A mechanism $f$ \textsl{sd-dominates} another mechanism $g$, denoted as $f\rhd^{sd} g$, if for any profile $R$ the random assignment
$f(R)$ stochastically $R$-dominates the random assignment $g(R)$, and for some profile $\bar{R}$ and $i\in N$ we have $f_i(\bar{R})\neq g_i(\bar{R})$. Further $f$ is \textsl{sd-efficient} if for all $R\in \mathcal{R}^N$, $f(R)$ is sd-efficient under $R$.
Similarly, we define ex-post (weak) efficiency, ex-post (weak) non-wastefulness, and individual rationality\footnote{For individual rationality its ex-ante and ex-post notions are equivalent.} for a mechanism. A mechanism $f$ is \textsl{deterministic} if for any profile $R$, $|supp(f(R))|=1$, i.e., the mechanism chooses one assignment with probability one. % A deterministic mechanism is \textsl{efficient} if it chooses an efficient assignment for any profile.

Then $f$ is \textsl{strategy-proof} if for all $R\in \mathcal{R}^N$, all $i\in N$ and all $R_i'\in \mathcal{R}^i$,
$f_i(R)$ stochastically $R_i$-dominates $f_i(R_i',R_{-i})$. Strategy-proofness is equivalent to the requirement that for any von Neumann-Morgenstern utility representation compatible with a given ordinal ranking of objects, submitting the true ordinal ranking maximizes an agent's expected utility. Most real-life mechanisms only elicit this ordinal information (instead of von Neumann-Morgenstern utilities).

Furthermore, $f$ is \textsl{envy-free} if for all $R\in \mathcal{R}^N$ and all $i\in N$,
$f_i(R)$ stochastically $R_i$-dominates $f_j(R)$ (where in $f_j(R)$ the outside option $j$ is replaced by $i$). If $f(R)$ attaches probability one to assignment $\mu$, then this is equivalent to $\mu_iR_i\mu_j$ for all $i,j\in N$.
Then $f$ satisfies \textsl{equal treatment of equals} if for all $R\in \mathcal{R}^N$ and all $i,j\in N$,
$R_i=R_j$ implies $f_{io}(R)=f_{jo}(R)$ for all $o\in O$.
Finally, $f$ satisfies \textsl{strong equal treatment of equals} (Nesterov, 2017) if for all $R\in \mathcal{R}^N$, all $i,j\in N$ and all $x\in O$,
$R_i|B(x,R_i)=R_j|B(x,R_j)$,\footnote{Note that this implies $B(x,R_i)=B(x,R_j)$.} then $f_{io}(R)=f_{jo}(R)$ for all $o\in B(x,R_i)$.
Note that \textsl{envy-freeness} implies \textsl{strong equal treatment of equals} which in turn implies \textsl{equal treatment of equals}.
There, \textsl{strong equal treatment of equals} is a fairness notion in between these two fairness notions.

%We also define two invariance conditions of a mechanism with respect to renaming agents and with respect to renaming objects.
%
%Given a permutation $\tau:N\rightarrow N$ and $R\in \mathcal{R}^N$, let $\tau(R)$ be the profile such that for all $i\in N$, %$\tau(R)_i=R_{\tau(i)}$. A mechanism $f$ is \textsl{anonymous} if for any permutation $\tau:N\rightarrow N$ and $R\in \mathcal{R}^N$, we have %$f_{i}(\tau(R))=f_{\tau(i)}(R)$ for all $i\in N$.
%
%Given a permutation $\sigma:O\rightarrow O$ and $R\in \mathcal{R}^N$, let $R_i^{\sigma}$ be such that (i) for all $a,b\in O$, $aR_ib$ iff %$\sigma(a)R_i^{\sigma}\sigma(b)$ and (ii) for all $a\in O$, $aR_ii$ iff $\sigma(a)R_ii$, and $R^{\sigma}=(R_i^{\sigma})_{i\in N}$. A mechanism %$f$ is \textsl{neutral} if for any permutation $\sigma:O\rightarrow O$ and $R\in \mathcal{R}^N$, we have $f_{io}(R)=f_{i\sigma(o)}(R^{\sigma})$ %for all $i\in N$ and all $o\in O$.

Most properties are defined in terms of agents' random assignments. For a given set of properties, we say that a mechanism $f$ is \textsl{unique in terms of probability shares}, if for any other mechanism $\phi$ satisfying this set of properties, $f(R)$ and $\phi(R)$ are welfare-equivalent for any profile $R$, i.e., if individual random assignments coincide. Below we introduce two well-known mechanisms.

%The uniform assignment (UA) mechanism\footnote{\cite{chambers2004consistency} characterizes UA via consistency.} randomizes uniformly over all $|N|!$ deterministic non-wasteful assignments (irrespective of agents preferences). Hence for individual object assignment probabilities we have: for all $R\in \mathcal{R}^N$,
%$U\!A_{io}(R)=\frac{1}{n}$ for all $i\in N$ and $o\in O$.

Let $\succ$ denote a strict priority ranking over $N$ and let $\Pi$ denote the set of all strict priority orders.
Given $\succ\in \Pi$, let $f^{\succ}$ denote the (deterministic) serial dictatorship (SD) mechanism where agents are assigned their most-preferred among all available acceptable objects in order of their priority, or remain unassigned if no acceptable objects remain.\footnote{For any $R\in \mathcal{R}^N$ and $i_1\succ i_2 \succ \cdots \succ i_n$, $i_1$ receives his most $R_{i_1}$-preferred element in $O\cup \{i_1\}$ (denoted by $f_{i_1}^{\succ}(R)$), and for $l=2,\ldots,n$, $i_l$ receives his most $R_{i_l}$-preferred object in $O\backslash \{f_{i_l}^{\succ}(R),\ldots,f_{i_{l-l}}^{\succ}(R)\}\cup\{i_l\}$ (denoted by $f_{i_l}^{\succ}(R)$).}
Then the random serial dictatorship (RSD) mechanism is defined by $R\!S\!D(R) =\frac{1}{n!}\sum_{\succ\in \Pi}f^{\succ}(R)$ for all $R\in \mathcal{R}^N$.%\footnote{\cite{pyciatroyan2023} construct a random mechanism distinct from RSD but where for any profile $R$ the chosen random assignment is equivalent to $RSD(R)$.}

We omit the formal definition of the probabilistic serial (PS) mechanism\footnote{For that, we refer the reader to Bogomolnaia and Moulin (2001).}
%; \cite{bogomolnaia2015random} offers an alternative definition of PS, and \cite{katta2006solution} extend PS to the domain where indifferences are allowed.}
and provide an intuitive formulation instead: agents start eating, with uniform speed, from their most-preferred acceptable object; once an object is exhausted, each agent eats with uniform speed from his most-preferred among the remaining acceptable objects, and so on until all of their acceptable objects are exhausted. The assignment probabilities of any agent in PS are simply the shares of objects the agent has eaten over the course of this process.\footnote{The PS-mechanism pins down individuals' object assignment probabilities directly but can be decomposed as a convex combination of deterministic assignments by the Birkhoff-von Neumann Theorem (Birkhoff, 1946).}

The literature widely discusses the tradeoff between these two mechanisms: on the one hand RSD satisfies ex-post efficiency, equal treatment of equals and strategy-proofness but violates sd-efficiency and envy-freeness while on the other hand PS satisfies sd-efficiency and envy-freeness but violates strategy-proofness.

\section{Agent-object-wastefulness}\label{sec4}

Let $f$ be a mechanism and $q\in [0,1]$. For any profile $R$ and any $o\in O$, let $f_{oo}(R)=1-\sum_{j\in N}f_{jo}(R)$ denote the unassigned share of object $o$ under $f(R)$.\newline

\noindent
\textbf{$q$-agent-object-wastefulness:}
$f$ is $q$-agent-object-wasteful if there exist a profile $R$,  $o\in O$, and $i\in N$ such that $oP_ii$ and $\min\{f_{ii}(R),f_{oo}(R)\}\geq q$.\newline

In words, a mechanism is $q$-agent-object-wasteful if we could feasibly increase the probability with which an object is assigned to an agent who desires it by $q$ and reduce the probability with which that agent and object remain unassigned by the same amount (without changing any other assignment probabilities). Furthermore, we call $f$ \textsl{at-most-$q$-agent-object-wasteful} if both $f$ is $q$-agent-object-wasteful and $f$ is not $q'$-agent-object-wasteful for $q'>q$, i.e., $q$ is the least upper bound of $f$ for agent-object-wastefulness.\footnote{Note that the set of agents, the set of assignments and the set of preference profiles are finite, i.e., for any mechanism \textsl{at-most-$q$-agent-object-wastefulness} is well-defined.}

Regarding $q$-agent-object wastefulness of RSD, for four agents and three objects we obtain for profile $G$ below:\footnote{Throughout rows stand for agents' preferences (where for instance, $G_1:ac$), the columns for objects, and the matrix entries for the random assignment chosen by the mechanism (for instance, $RSD_{3a}(G)=\frac{5}{24}$).}
\begin{equation}\label{RSD1}
\begin{array}{cccc}
RSD(G) & a & b & c \\ \hline
ac & \frac{14}{24} & 0 & \frac{6}{24} \\
bc & 0 & \frac{1}{3} &  \frac{14}{24}\\
ba & \frac{5}{24} & \frac{1}{3} & 0 \\
ba & \frac{5}{24} & \frac{1}{3} &  0
\end{array}.
\end{equation}
%Objects $a$ and $b$ are fully assigned as they are top ranked.
There are 24 orders and object $c$ is unassigned for the orders 1234, 1243, 2134, 2143, i.e., with probability $\frac{4}{24}=\frac{1}{6}$.
Agent 1 is unassigned for the orders 4321, 4231, 3421 and 3241, i.e., with probability $\frac{4}{24}=\frac{1}{6}$.
%Note that agent 2 is unassigned for the orders 4312 and 3412.
Hence, RSD is $\frac{1}{6}$-agent-object-wasteful.

We introduce below a new fairness notion in between \textsl{strong equal treatment of equals} and \ETE.

\begin{definition}
Let $R$ be a profile and $p$ be a random assignment.
Then $p$ satisfies \emph{(almost) equal treatment of almost equals (ETAE)} if for all $i,j\in N$ such that (i)
$R_i|O=R_j|O$ and (ii) $|A(R_i)|\leq |A(R_j)|\leq |A(R_i)|+1$ we have $p_{iy}= p_{jy}$ for all $y\in A(R_i)$.
\end{definition}

In words, if the preferences of two agents over objects are identical and their preferences differ by at most one acceptable object, then the probability share of any common acceptable object is identical. Note that then $R_i$ and $R_j$ differ by at most one adjacent switch between an object and being unassigned.\footnote{This means that the Kemeny distance (measuring the minimal number of pairwise switches from $R_i$ to obtain $R_j$) between $R_i$ and $R_j$ is at most one.} In this sense it is the weakest conceivable strengthening of \ETE.\footnote{Furthermore, on the domain where all objects are acceptable, \SETE\ and \ETE\ are equivalent, while on the full domain \SETE\ implies \ETE\ but not the other way around (and \SETE\ is weaker than \textsl{strong equal treatment of equals}). In addition, when all agents have equal priorities, (i) it is straightforward to check that no ex-ante discrimination (Kesten and \"Unver, 2015) implies \SETE\ and that ordinal fairness (Han, 2024b) implies \SETE.} Both PS and RSD satisfy \SETE.\footnote{For PS this follows as agents eat at uniform speed following the order of their objects while for RSD this follows as it satisfies ``bounded invariance'' (whereby changing the set of acceptable objects while keeping the same ranking over objects does not alter the distributions of the objects which are acceptable before and after the change), and RSD satisfies \ETE\ and \SP.}

We show for four or more agents and for three or more objects that $\frac{1}{6}$-agent-object-wastefulness is not specific for RSD as it is attained (or surpassed) by any mechanism satisfying (i) \SETE, \SP\ and \XWNW, and respectively, (ii) \ETE, \SP\ and \XEFF.

\begin{theorem}\label{the1}
Let $|N|\geq 4$ and $|O|\geq 3$.
\begin{itemize}
\item[(i)]
If $f$ satisfies \SETE, \SP\ and \XWNW, then $f$ is $q$-agent-object-wasteful with $q\geq \frac{1}{6}$.
%\item[(ii)] If $f$ satisfies \ETE, \WEF, \SP\ and \XWNW, then $f$ is $q$-agent-object-wasteful with $q\geq \frac{1}{6}$.
\item[(ii)] If $f$ satisfies \ETE, \SP\ and \XEFF, then $f$ is $q$-agent-object-wasteful with $q\geq \frac{1}{6}$.
\end{itemize}
\end{theorem}

The proof of Theorem \ref{the1} and Lemma \ref{lemm1} below are relegated to the Appendix. We further show later that Theorem \ref{the1} does not remain true for four agents and three objects when \SETE\ is weakened to \ETE\ as random deferred acceptance (RDA) is at-most-$q$-agent-object-wasteful with $q<\frac{1}{6}$.

RSD indeed obtains the minimal bound of $\frac{1}{6}$ for $q$-agent-object-wastefulness when there are four agents or three objects.\footnote{For fewer agents or objects, RSD is at-most-$0$-agent-object-wasteful (Martini, 2016, Proposition 1 $\&$ 2).}

\begin{lemma}\label{lemm1}
\begin{itemize}
\item[(i)]
Let $|N|\geq 4$ and $|O|=3$.
Then RSD is at-most-$\frac{1}{6}$-agent-object-wasteful.
\item[(ii)]
Let $|N|= 4$ and $|O|\geq 3$.
Then RSD is at-most-$\frac{1}{6}$-agent-object-wasteful.
\end{itemize}
\end{lemma}

As RSD satisfies \SETE, \SP\ and \XWNW,
we obtain the following from Theorem \ref{the1} and Lemma \ref{lemm1}.

\begin{corollary}\label{COR1}
For $|N|= 4$ or $|O|= 3$, RSD attains the minimal bound for at-most-$q$-agent-object-wastefulness  
\begin{itemize}
\item[(i)] in the class of
mechanisms satisfying \SETE, \SP\ and \XWNW, and 
\item[(ii)] in the class of mechanisms satisfying \ETE, \SP\ and \XEFF.
\end{itemize}
\end{corollary}

Note that Corollary \ref{COR1} covers the case $|N|=|O|\leq 4$ for which all impossibility results in the literature of random assignment were shown.\footnote{See for instance, Bogomolnaia and Moulin (2001), Martini (2016), Nesterov (2017), Basteck (2018) and Basteck and Ehlers (2023).} Moreover, the class of mechanisms satisfying \SETE, \SP\ and \XEFF\ includes mechanisms that are strictly more agent-object-wasteful.

\begin{proposition}
There exist mechanisms which are $\frac{5}{24}$-agent-object-wasteful and satisfy \SETE, \SP\ and \XEFF.
\end{proposition}

The proof is by example in Appendix \ref{Erdil reconsidered}.

\section{Object-wastefulness}\label{sec5}

An object, which remains unassigned with positive probability, may be desired by multiple
agents. Thus, instead of $q$-agent-object-wastefulness, we also consider $q$-object-wastefulness,
defined as follows:\newline

\noindent
\textbf{$q$-object-wastefulness:}
$f$ is $q$-object-wasteful if there exists a profile $R$ and $o\in O$ such that
$\min\{\sum_{i\in N:oP_ii}f_{ii}(R),f_{oo}(R)\}\geq q$.\newline

In words, a mechanism is $q$-object-wasteful if we could feasibly increase the probability with which an object is assigned to agents who desire it by $q$ and reduce the probabilities with which the agents and the object remain unassigned by the same amount. Similarly, as above we then define \textsl{at-most-$q$-object-wastefulness} for a mechanism. Note that $q$-agent-object-wastefulness implies $q$-object-wastefulness.

Regarding $q$-object-wastefulness of RSD, for four agents and three objects we obtain for profile $C''$ below:
\begin{equation}\label{RSD2}\begin{array}{cccc}
RSD(C'') & a & b & c \\ \hline
abc & \frac{1}{2} & \frac{3}{8} &  0\\
cba & 0 & \frac{3}{8} &  \frac{1}{2}\\
a & \frac{1}{2} & 0 &  0\\
c & 0 & 0 &  \frac{1}{2}
\end{array}.
\end{equation}
There are 24 orders and $b$ is unassigned for the orders 1234, 1243, 2134, 2143, 1324 and 2413, i.e., with probability $\frac{6}{24}=\frac{1}{4}$. Agent 1 is unassigned for the orders 4321, 4231 and 3421, i.e., with probability $\frac{3}{24}=\frac{1}{8}$ (and similarly for agent 2). Hence, RSD is $q$-object-wasteful with $q \geq \frac{1}{4}$.

%We define below \textsl{weak envy-freeness} for the full domain whereby agents are allowed to freely dispose probability shares.
%
%\begin{definition}
%Let $R$ be a profile and $p$ be a random assignment. Then $p$ is \textsl{weakly envy-free} if for any $i,j\in N$, there exists no $\hat{p}_j$ %such that $p_j$ object-by-object dominates $\hat{p}_j$, $\hat{p}_j\neq p_i$ and $\hat{p}_j$ $R_i$-dominates $p_i$.
%\end{definition}
%This means that $i$ can possibly dispose of probability shares of $j$'s random assignment $p_j$ to obtain a better random assignment than $p_i$ %in terms of stochastic dominance. Note that there is no relation among \SETE\ and \WEF.\footnote{Note that \SETE\ only applies to two agents %having the same ranking over objects (and thus does not imply \WEF) while for $O=\{a,b,c\}$, $R_1:abc$, $R_2:ab$, %$p_1=(p_{1a},p_{1b},p_{1c},p_{11})=(0.2,0.1,0.1,0.6)$ and
%$p_2=(p_{2a},p_{2b},p_{2c},p_{22})=(0.3,0,0,0.7)$ \WEF\ is satisfied and \SETE\ is violated.}

We show for four or more agents and for three or more objects that $\frac{1}{4}$-object-wastefulness is not specific for RSD as it is attained (or surpassed) by any mechanism satisfying (i) \SETE, \SP\ and \XWNW, %(ii) \ETE, \WEF, \SP\ and \XWNW,
or (ii) \ETE, \SP\ and \XEFF.
%Hence, the common lower bound of $\frac{1}{4}$ is obtained for these two classes of mechanisms regarding $q$-object wastefulness.

\begin{theorem}\label{theo2}
Let $|N|\geq 4$ and $|O|\geq 3$.
\begin{itemize}
\item[(i)] If $f$ satisfies \SETE, \SP\ and \XWNW, then $f$ is $q$-object-wasteful with $q\geq \frac{1}{4}$.
%\item[(ii)] If $f$ satisfies \ETE, \WEF, \SP\ and \XWNW, then $f$ is $q$-object-wasteful with $q\geq \frac{1}{4}$.
\item[(ii)] If $f$ satisfies \ETE, \SP\ and \XEFF, then $f$ is $q$-object-wasteful with $q\geq \frac{1}{4}$.\footnote{Basteck and Ehlers (2025) established on the domain, where all objects are acceptable, that RSD is not characterized by these three properties, i.e., already on this domain there is a full class of mechanisms satisfying those three properties.}
\end{itemize}
\end{theorem}

Furthermore, for four or more agents and three objects, RSD attains the minimal bound of $\frac{1}{4}$ for $q$-object-wastefulness.
%there is no difference whether we regard $q$-object-wastefulness or the strongest notion of $q$-wastefulness as they give us the same results.
%\begin{remark}
%Note that for three objects, for any mechanism satisfying ETE, SP and XEFF, any other aggregate waste-notion will give the same results as at %most one object is wasted, i.e., it suffices to consider (aggregate)$q$-agent-object-wastefulness.
%\end{remark}

\begin{lemma}\label{lem.RSD.tight}
For $|N|\geq 4$ and $|O|=3$, RSD is at-most-$\frac{1}{4}$-object-wasteful.
\end{lemma}

As RSD satisfies \SETE, \SP, and \XEFF,
we obtain the following from Theorem \ref{theo2} and Lemma \ref{lem.RSD.tight}.

\begin{corollary}\label{cor.RSD.tight}
For $|N|\geq  4$ and $|O|=3$, RSD attains the minimal bound of $\frac{1}{4}$ for at-most-$q$-object-wastefulness
\begin{itemize}
\item[(i)] in the class of mechanisms satisfying \SETE, \SP\ and \XWNW, and
%\item[(ii)] the class of mechanisms satisfying \ETE, \WEF, \SP\ and \XWNW, and
\item[(ii)] in the class of mechanisms satisfying \ETE, \SP\ and \XEFF.
\end{itemize}
\end{corollary}

\section{Random Deferred Acceptance}\label{secRDA}

We first define the (agent-proposing) deferred-acceptance (DA) algorithm.
A priority structure $\gtrsim=(\gtrsim_o)_{o\in O}$ specifies a (strict) priority order $\gtrsim_o$ for each object $o$. Now for $\gtrsim$ we run the DA-algorithm with profile $R$ and denote its outcome by $DA^{\gtrsim}(R)$.

\begin{quote}
\begin{center}
\textbf{DA Algorithm}
\end{center}

Let $\gtrsim=(\gtrsim_o)_{o\in O}$ and profile $R$ be given.

\textbf{Step 1.} Each agent proposes to his top-ranked acceptable object. If there is no such object, then he is unassigned. Each object $o$ considers the set of proposals that it receives. Among them, it tentatively accepts the highest $\gtrsim_o$-ranked agent and rejects the others. If there is no rejection, then stop. Otherwise, move to Step $2$.

\textbf{Step} {\boldmath $t\geq 2$.} Each agent who is rejected at Step $t-1$ proposes to his top-ranked acceptable object among the ones that have not rejected him yet. If there is no such object, then he is unassigned. Each object $o$ considers the agent whom it tentatively accepted at Step $t-1$ together with agents who have proposed at Step $t$. Among them, it tentatively accepts  the highest $\gtrsim_o$-ranked agent and rejects the others. If there is no rejection, then stop. Otherwise, move to Step $t+1$.
\end{quote}

We denote the outcome of the (agent-proposing) DA-algorithm for $\gtrsim$ and $R$ by $DA^{\gtrsim}(R)$.

The random DA mechanism (RDA) puts probability $\left(\frac{1}{|N|!}\right)^{|O|}$ on any priority structure $\gtrsim$ and its outcome $DA^{\gtrsim}(R)$ (where for identical outcomes probabilities are added). RDA satisfies \ETE, \SP\ and \XNW\ but violates both \SETE\ and \XEFF.
RDA violates \SETE\ as for $N=\{1,2,3\}$, $O=\{a,b\}$ and profile $R'$ below we obtain:\footnote{Note that 1 obtains $a$ with probability $\frac{1}{2}$ when $1> _a2$, 1 obtains $b$ with probability $\frac{1}{4}$ when $2>_a 1$ and $1>_b3$, while $3$ obtains $a$ with probability $\frac{1}{12}$ when $3>_a2>_a1$ (which arises with probability $\frac{1}{6}$) and $1>_b3$.}
\[
\begin{array}{cccc}
RDA(R') & a & b  \\ \hline
ab & \frac{1}{2} & \frac{1}{4}  \\
a & \frac{5}{12} & 0 \\
ba & \frac{1}{12} & \frac{3}{4}
\end{array}.
\]
Hence, $RDA_{1a}(R')>RDA_{2a}(R')$ and $RDA(R')$ violates \SETE.\footnote{One may also view it as a violation of weak (sd-)envy-freeness in the sense that, restricted to objects that $2$ finds acceptable, $1$ receives a lottery over objects that stochastically dominates the lottery of $2$. }
% and \WEF\ as agent 2 can dispose $RDA_{1b}(R')$ to obtain a random assignment which sd-dominates $RDA_2(R')$.

For profile $G$ in (\ref{RSD1}), under RDA objects $a$ and $b$ are fully assigned whereas $c$ remains unassigned with probability $\frac{1}{9}$ (as $c$ is unassigned only if $2$ receives $b$ (which arises when both $2>_b3$ and $2>_b4$, i.e., with probability $\frac{1}{3}$) and 1 receives $a$ (which arises when both $1>_a3$ and $1>_a4$, i.e., with probability $\frac{1}{3}$).

Hence, $RDA(G)$ is $\frac{1}{9}$-agent-object wasteful as\footnote{Note that 1 is assigned $a$ when $1>_a\{3,4\}$ which arises with probability $\frac{1}{3}$, when both $3>_b\{2,4\}$ and $3>_a1>_a4$ which arises with probability $\frac{1}{3}\cdot \frac{1}{6}=\frac{1}{18}$, and
when both $4>_b\{2,3\}$ and $4>_a1>_a3$ which arises with probability $\frac{1}{3}\cdot \frac{1}{6}=\frac{1}{18}$. Hence, $RDA_{1a}(G)=\frac{4}{9}$. Furthermore, 1 is unassigned only when 2 receives $c$, and then either 3 receives $b$ and $4$ $a$ (which arises when $3>_b\{2,4\}$, $4>_a1$ and $2>_c1$, i.e., with probability $\frac{1}{12}$ or 3 receives $a$ and $4$ $b$ (which arises also with probability $\frac{1}{12}$). Hence, $RDA_{11}(G)=\frac{1}{6}=\frac{3}{18}$. %Furthermore, 2 is assigned $c$ when 3 is assigned $b$ and 1 is assigned $a$, i.e., when $3>_b\{2,4\}$ and $1>_a4$ which arises with probability $\frac{1}{3} \cdot \frac{1}{2}=\frac{1}{6}$, and 2 is assigned $c$ when 4 is assigned $b$ and 1 is assigned $a$, i.e., when $4>_b\{2,3\}$ and $1>_a3$ which arises with probability $\frac{1}{3} \cdot \frac{1}{2}=\frac{1}{6}$. Hence, $RDA_{2c}(G)=\frac{1}{3}$.
Now the remaining entries in $RDA(G)$ follow from \ETE, feasibility and the fact that $c$ is wasted with probability $\frac{1}{9}$ while $a$ and $b$ are fully assigned.}
\begin{equation}\label{RDA2}
\begin{array}{cccc}
RDA(G) & a & b & c \\ \hline

ac & \frac{4}{9} & 0 & \frac{7}{18} \\

bc & 0 & \frac{1}{3} &  \frac{1}{2}\\

ba & \frac{5}{18} & \frac{1}{3} & 0 \\

ba & \frac{5}{18} & \frac{1}{3} &  0
\end{array},
\end{equation}
$RDA_{22}(G)=\frac{1}{6}$, $cP_22$ and $RDA_{cc}(G)=\frac{1}{9}$. Note that then from (\ref{RSD1}) it follows that
$RSD_2(G)$ stochastically $G_2$-dominates $RDA_2(G)$ while $RDA_3(G)$ stochastically $G_3$-dominates $RSD_3(G)$,
%\footnote{This even remains true under the ``bi-dominance'' relation considered by Zhang (2023b) whereby a random assignment both strictly sd-dominates and strictly size-dominates another random assignment.}
i.e., there is no stochastic dominance relation among $RSD$ and $RDA$.

While the above establishes lower agent-object-wastefulness of RDA compared to
RSD at profile $G$, below we establish bounds for RDA with respect to $q$-agent-object-wastefulness
and $q$-object-wastefulness across all preference profiles.
%Below we establish bounds for RDA with respect to $q$-agent-object-wastefulness and $q$-object-wastefulness.
% (and its proof is relegated to the Appendix).

\begin{lemma}\label{lemRDA}
Let $|N|=4$ and $|O|=3$. Then the following holds.
\begin{itemize}
\item[(a)] RDA is at-most-$q$-agent-object-wasteful with $q< \frac{1}{6}$.
\item[(b)] RDA is at-most-$q$-object-wasteful with $q\geq \frac{1}{4}$.
\end{itemize}
\end{lemma}

For four agents and three objects, on the one hand, by Lemma \ref{lemRDA} (a) there exist mechanisms in the class of \ETE, \SP\ and \XNW\ having at-most-$q$-agent-object waste with $q<\frac{1}{6}$. On the other hand, Theorem \ref{the1} implies that any such mechanism must violate
\SETE.%, weak envyfreeness and \XEFF.

While for RDA at-most-$q$-agent-object-waste is smaller than for RSD, for RDA and RSD $q$-object-waste remains greater than or equal to $\frac{1}{4}$. In fact, for profile $\hat{R}$ in the proof of Lemma \ref{lemRDA} (b) we obtain that $RDA(\hat{R})=RSD(\hat{R})$ (where for determining $RDA(\hat{R})$ we use the same arguments as in (\ref{RSD2})). In particular, the mixture mechanisms $\alpha RDA + (1-\alpha) RSD$ (where $0<\alpha<1$) are $\frac{1}{4}$-object wasteful %while having at-most-$q$-agent-object waste with $q<\frac{1}{6}$.
while having, for four agents and three objects, at-most-$q$-agent-object waste with $q < \frac{1}{6}$.\footnote{If one agent ranks no object acceptable, then from Martini (2016, Proposition 2) RSD is 0-agent-object-wasteful and this follows from Lemma 3 (a); if two objects are top ranked or two agents rank them in the same order first and second (say $a$ and $b$), then only the third object $c$ is possibly not fully assigned under RSD and RDA and this follows from Lemma 3 (a) and  Lemma 1; and otherwise, all agents have the same top object, say $a$, and exactly two agents rank $b$ (say 1) and $c$ (say 2) second, but then under both RSD and RDA 1 and 2 are never unassigned.}

Below we discuss single-tie-breaking versus multiple-tie-breaking for RDA when all agents have equal priority at all objects.

\begin{remark}
Suppose that there are at least four agents, at least three objects, and all agents have equal priority for all objects. Then STB-RDA (where STB stands for single-tie-breaking) is equal to RSD whereas MTB-RDA (where MTB stands for multiple tie-breaking) is equal to RDA.
\begin{itemize}
\item[(i)] On the one hand, if multiple tie-breaking is deterministic, then Abdulkadiro\u{g}lu, Pathak and Roth (2009, Theorem 1) show that no strategy-proof deterministic mechanism dominates
$DA^{\gtrsim}$. Note that single-tie-breaking is the instance of multiple-tie-breaking where all tie-breakers are identical for all objects.
On the other hand, if single-tie-breaking is random, then STB-RDA becomes RSD which is sd-dominated by a strategy-proof random mechanism by Erdil (2014, Proposition 3 (i)).\footnote{Note that Erdil (2014, Proposition 3 (ii)) only considers STB-RDA.}

But now it follows that the same conclusion holds for multiple-tie-breaking as we may write
\[ RDA = \alpha RSD + (1-\alpha) h\]
where $\alpha$ denotes the probability for MTB having all tie-breaking orders to be identical\footnote{As there are $n!$ tie-breaking orders, we have $\alpha= \frac{|N|!}{(|N|!)^{|O|}}=\frac{1}{(|N|!)^{|O|-1}}$.} whereas $h$ denotes the part of MTB-RDA where at least two tie-breaking-orders are distinct. Now let $\bar{h}$ denote the mechanism constructed in the proof of Erdil (2014, Proposition 3 (i)).
But then $\alpha \bar{h} + (1-\alpha) h \rhd^{sd}  RDA$.

Therefore, on the one hand MTB-DA is not dominated by a strategy-proof deterministic mechanism while there exist strategy-proof random mechanisms sd-dominating MTB-RDA and STB-RDA.\footnote{Han (2024a) determines conditions under which STB-RDA is constrained efficient.}
\item[(ii)] For RDA at-most-$q$-agent-object-waste is smaller than for RSD for four agents and three objects. %, for RDA maximal $q$-object-waste is greater than or equal to the one of RSD (which equals $\frac{1}{4}$).
    Note that for receiving top objects Abdulkadiro\u{g}lu, Pathak and Roth (2009, Table 3) make a point in favor of STB-DA over MTB-DA by random simulations\footnote{Other contributions comparing top choices for (random) tie-breaking are Arnosti (2023), Ashlagi and Nikzad (2020), and Ashlagi, Nikzad and Romm (2019).} whereas we established a criterion in favor of MTB-DA over STB-DA in terms of at-most-$q$-agent-object-wastefulness.
\end{itemize}
\end{remark}

%\section{Discussion}

\section{Ex-Ante (Weak) Non-Wastefulness}\label{sec6}

%Note that our ex-ante notions are the ones corresponding to ``weak non-wastefulness'' as we only consider waste whereby agents are unassigned %with positive probability. This is weaker than the ex-ante notions of non-wastefulness of Erdil (2014) and Martini (2016) whereby agents are %assigned with positive probability less preferred objects.

Erdil (2014) introduced the following ex-ante notion of non-wastefulness: a random assignment is ex-ante non-wasteful if there is no agent-object pair such that with positive probability both the object is unassigned and the agent receives a less preferred object. For our purposes, we also introduce the ex-ante notion of weak non-wastefulness whereby the agent is unassigned with positive probability and ranks the object acceptable.\newline

\noindent
\textbf{Ex-ante non-wastefulness:} For all $R\in \mathcal{R}^N$, there exist no $i\in N$ and $o,\hat{o}\in O\cup\{i\}$ such that $f_{oo}(R)>0$, $oP_i\hat{o}$ and $f_{i\hat{o}}(R)>0$.\newline

\noindent
\textbf{Ex-ante weak non-wastefulness:} For all $R\in \mathcal{R}^N$, there exist no $i\in N$ and $o\in O$ such that $f_{oo}(R)>0$, $oP_ii$ and $f_{ii}(R)>0$.\newline

It is obvious that these ex-ante notions imply the corresponding weaker ex-post ones, i.e., ex-ante non-wastefulness implies ex-post non-wastefulness and ex-ante weak non-wastefulness implies ex-post weak non-wastefulness. Furthermore, ex-ante weak non-wastefulness is equivalent to at-most-0-agent-object-wastefulness.%, which in turn is equivalent to aggregate-0-object-wastefulness.

Martini (2016) has shown the impossibility of \ETE, \SP\ and ex-ante non-wastefulness. From our first main theorem we obtain the following result (as ex-ante weak non-wastefulness implies ex-post weak non-wastefulness and $\frac{1}{6}$-agent-object-wastefulness contradicts ex-ante weak non-wastefulness).

\begin{corollary}\label{exanteWNW}
Let $|N|\geq 4$ and $|O|\geq 3$. Then %the following holds:
%\begin{itemize}\item[(i)]
there exists no mechanism satisfying \SETE, \SP\ and \textsl{ex-ante weak non-wastefulness}.
%\item[(ii)] There exists no mechanism satisfying \ETE, \SP, \XEFF\ and \textsl{ex-ante weak non-wastefulness}.
%\end{itemize}
\end{corollary}

Note that in comparison to Martini (2016) the above uses the stronger fairness notion of \SETE\ and at the same time the weaker ex-ante notion of ex-ante weak non-wastefulness. Hence, neither Martini's (2016) impossibility implies Corollary \ref{exanteWNW} nor the other way around.\footnote{The (in)compatibility of \ETE, \SP\ and ex-ante weak non-wastefulness remains an open question. The compatibility of \SETE, \SP, and \XEFF\ (which implies \XNW\ and \XWNW) follows as RSD satisfies all those requirements.}

\begin{remark}\rm
Bogomolnaia and Moulin (2015) analyze maximizing the number of assigned objects as an important concern in the class of \textsl{envy-free} mechanisms. They show by considering ``diagonal profiles'' where PS obtains the worst ratio of assigned objects relative to the number of agents, that any envy-free mechanism has a weakly larger worst ratio than PS.
\end{remark}
\begin{remark}\rm
Regarding the class of \textsl{envy-free}, \textsl{strategy-proof} and \textsl{ex-post weakly non-wasteful} mechanisms, 
%Basteck and Ehlers (2023) and Shende and Purohit (2023) have shown for the domain, where all objects are acceptable, with equal numbers of agents and objects that any such mechanism is ex-post efficient with probability at most $\frac{2}{|N|}$. 
Basteck and Ehlers (2023) have shown, for the domain where all objects are acceptable and where there are equal numbers of agents and objects, that any 
such mechanism is ex-post efficient with probability at most  $\frac{2}{|N|}$ (which include the rank exchange mechanisms by Shende and Purohit (2023)).
Thus, under these criteria, ex-post inefficiency for any such mechanism converges to one as $|N|$ grows.
Note that RSD violates envy-freeness and here we use the much weaker requirement of \SETE.\footnote{Basteck (2018), Demeulemeester and Pereyra (2024) and Duddy (2025) consider other ``fairness'' and ``egalitarian'' criteria than envy-freeness whereas Heo, Manjunath and Alva (2025) consider ``unambiguous efficiency''.}
\end{remark}

\section{Conclusion}\label{sec7}

The ``strongest'' non-wastefulness notion allows unassigned objects to be distributed optimally in order to minimize the aggregate probability of being unassigned (while object-by-object dominating the original one in terms of probability shares).
Given two random assignments $p$ and $\bar{p}$, we say that $p$ \textit{object-by-object dominates} $\bar{p}$ if for all $i\in N$ and all $o\in O$ we have $p_{io}\geq \bar{p}_{io}$ (with strict inequality holding for some $j\in N$ and some $x\in O$). Object-by-object domination makes no reference to a preference profile. Now given profile $R$ and two individually rational random assignments $p$ and $\bar{p}$, if $p$ object-by-object dominates $\bar{p}$, then $p$ stochastically $R$-dominates $\bar{p}$ and $p_i\neq \bar{p}_i$ for some $i\in N$. Note that for this assertion we do not need to know the exact preferences over acceptable objects.%\footnote{On the full domain we show later in Lemma \ref{obyo} that if one strategy-proof mechanism sd-dominates another strategy-proof mechanism, then we can find a profile where object-by-object dominance holds.}

Formally, given profile $R$,
for any $p\in \mathcal{IR}(R)$ object-by-object dominating $\bar{p}$ let
\[ \Delta(p,\bar{p})=\sum_{i\in N}\sum_{o\in A(R_i)}(p_{io}-\bar{p}_{io}).\]
Let $f$ be an arbitrary mechanism.\newline

\noindent
\textbf{$q$-wastefulness:} $f$ is $q$-wasteful if there exist a profile $R$ and $p\in \mathcal{IR}(R)$ object-by-object dominating $f(R)$ such that
%\[\max_{p\in \mathcal{IR}(R): p \, object-by-object\, dominates\, f(R)} \Delta(p,f(R))\geq q.\]
\[\Delta(p,f(R))\geq q.\]

In words, if we can reduce the aggregate probability of being unassigned by $q$ while object-by-object dominating the initial random assignment.
%Obviously, if $f$ is $q$-object-wasteful or $q$-agent-wasteful, then $f$ is $q$-wasteful.
%Similarly, as above we then define \emph{at-most-$q$-wastefulness} for a mechanism.

If both $|N|$ and $|O|$ become large, then $q$-waste becomes unbounded both in the class of \SETE, \SP\ and \XEFF\ mechanisms
as (\ref{RSD2}) can be replicated an arbitrary number of times (by partitioning the problem into subproblems of four agents and three objects). Hence, the notion of $q$-waste is not meaningful in the large.

For the domain where all agents rank all objects acceptable, ex-post weak non-wastefulness implies ex-ante weak non-wastefulness (as for $|N|\leq |O|$ no agent is unassigned and for $|N|\geq |O|$ no object is unassigned) and the issue of waste does not arise. While normative characterizations have been recently obtained regarding RSD by Pycia and Troyan (2026) and Basteck (2024) for these domains, we are the first ones to obtain positive support for RSD in terms of waste measures when agents are allowed to rank objects unacceptable. Although ex-ante notions of waste are impossible to achieve, RSD obtains the minimal bound of $\frac{1}{6}$ for agent-object-wastefulness and the minimal bound of $\frac{1}{4}$ for object-wastefulness.
Waste is an important concern in market design applications such as school choice and organ exchange where outside options are represented by private schooling and the compatibility with the own donor. At the same time it is socially preferable for the market to include those participants by ensuring them at least their outside option.

%Example in Erdil (2014, Proposition 4): RSD is $\frac{1}{12}$-agent-object-wasteful and $\frac{1}{6}$-object-wasteful (not maximal bounds).
%stress importance of outside options

\appendix

\begin{appendix}
\begin{center}
\textbf{APPENDIX.}
\end{center}

Throughout the Appendix we use the abbreviations ETE for \ETE, ETAE for \SETE, %WEF for \WEF,
SP for \SP, XWNW for \XWNW, XNW for \XNW\ and XEFF for \XEFF.

Given profile $R$ and $i\in N$, let $top(R_i)$ denote the most $R_i$-preferred object and $top(R)=\cup_{i\in N}\{top(R_i)\}$ denote the set of objects which are ranked at the top under $R$.

\section{Proof of Theorem \ref{the1}}

We begin with some preliminary observations. The first lemma lists conditions under which objects have to be split uniformly between agents ranking them at the top.

\begin{lemma}\label{lem.equal.split.at.top}
Let $|N|=4$, $|O|=3$, and $f$ satisfy SP, ETE, and XWNW.
\begin{itemize}
\item[(a)] For any profile $R$ where all agents rank $x$ at the top, $f_{ix}(R)=\frac{1}{|N|}$ for all $i\in N$.
\item[(b)] If $f$ also satisfies ETAE or XEFF, then for any profile $R$ with $top(R_i)\neq i$ for all $i\in N$ and where all agents in $S\subseteq N$ rank $x$ at the top while for all distinct $j,k\in N\backslash S$, $x\neq top(R_j)\neq top(R_k)\neq x$, then $f_{ix}(R)=\frac{1}{|S|}$ for all $i\in S$.
\item[(c)] If $f$ also satisfies ETAE or XEFF, then for any profile $R$ with $top(R_i)\neq i$ for all $i\in N$ and where all agents in $S= N\backslash\{j\}$ rank $x$ and $y$ at the top and in the same order while for $j\in N\backslash S$, $top(R_j)=z$, then $f_{ix}(R)=\frac{1}{|S|}$ and $f_{iy}(R)=\frac{1}{|S|}$ for all $i\in S$.

\end{itemize}
\end{lemma}
\noindent\textbf{Proof.}
Let $N=\{1,2,3,4\}$ and $O=\{a,b,c\}$.

\underline{\emph{(a):}} As $|O|<|N|$, some agent is unassigned with positive probability. Thus, if all agents consider $x$ acceptable, XWNW implies that $x$ is assigned with probability~1. Now, towards a contradiction, assume that there is a profile $R$ where all rank $x$ at the top, yet $x$ is not split uniformly. Further, w.l.o.g., let $R$ be minimal in that there is no such profile with some agent declaring fewer objects acceptable (and no agent declaring more objects acceptable). Now, take $j$ for whom $f_{jx}(R)\neq \frac{1}{|N|}$. Since $R$ is minimal, SP implies that $j$ considers only $x$ acceptable. If $f_{jx}(R)>\frac{1}{|N|}$, there must be another agent $k$ for whom $f_{kx}(R)<\frac{1}{|N|}$ and by ETE he must rank another object acceptable below $x$ -- but then there is another profile $R'$ where $k$ considers only $x$ acceptable, thus ranks fewer objects acceptable, and where by SP $f_{kx}(R')=f_{kx}(R)<\frac{1}{|N|}$, contradicting our assumption that $R$ was minimal. Analogously one finds a contradiction if $f_{jx}(R)<\frac{1}{|N|}$.
%Suppose that for some $i\in N$, $f_{ix}(R)\neq \frac{1}{|N|}$, say $i=1$. Let $j>1$ be minimal such that $R_j\neq R_1$. By ETE, $f_{1x}(R)=\cdots=f_{j-1x}(R)\neq \frac{1}{|N|}$. Then there exists $k\geq j$ such that $f_{kx}(R)\neq \frac{1}{|N|}$. Let $R'=(R_1,R_{-k})$, i.e., $k$ reports $R_1$. Now by SP and ETE, $\frac{1}{|N|}\neq f_{kx}(R)=f_{kx}(R')=f_{1x}(R')=\cdots=f_{j-1x}(R')\neq \frac{1}{|N|}$. Thus, there exists $l\in N\backslash \{1,\ldots,j-1,k\}$ such that $f_{lx}(R')\neq \frac{1}{|N|}$. Now continuing in this way we find a contradiction using ETE, SP and the fact that $x$ is assigned with probability 1.

\underline{\emph{(b):}} Let $R$ be a profile with $top(R_i)\neq i$ for all $i\in N$ and where all agents in $S\subseteq N$ rank~$x$ at the top while for all distinct $j,k\in N\backslash S$, $x\neq top(R_j)\neq top(R_k)\neq x$. As $|O|=3$, we must have $|S|\geq 2$. If $|S|=4$, then (a) yields the desired conclusion. If $|S|<4$ and $f$ satisfies XEFF, then $x$ must be split among agents in $S$ and we can use the same argument as in (a) to show that it is split uniformly.

Next, let $f$ satisfy ETAE and let $|S|=3$, say $S=\{1,2,3\}$. Suppose to the contrary that for some $i\in S$, say $i=1$, we have $f_{1x}(R)< \frac{1}{3}$. Let $R_1'$ be such that $R_1'|O=R_2|O$ and $|A(R_1')|=2$ and $R_2'$ be such that $R_2'|O=R_2|O$ and $A(R_2')=\{x\}$. Then
\[ \frac{1}{3}> f_{1x}(R)=f_{1x}(R_1',R_{-1})=f_{2x}(R_1',R_{-1})=f_{2x}(R_1',R_2',R_{-1,2})=f_{1x}(R_1',R_2',R_{-1,2}),\]
where the first and third equality follow from SP while the second and fourth equality follow from ETAE. For the next step, set $R_1''=R_2'$ (so that both consider only $x$ acceptable) and find by SP and ETE that
\[ \frac{1}{3}> f_{1x}(R_1'',R_2',R_{-1,2})=f_{2x}(R_1'',R_2',R_{-1,2}).\]
Further, we have $f_{3x}(R_1'',R_2',R_{-1,2})\leq \frac{1}{3}$ as otherwise $3$ may report $R_3'=R_1''=R_2'$ and we obtain a contradiction to feasibility using SP and ETE. If $|A(R_3)|>1$, then we may repeat the same argument as above and let agent 1 report $R_1'''$ such that $R_1'''|O=R_3'|O$ and $|A(R_1''')|=2$ while $3$ shortens the list of acceptable objects to $\{x\}$ (as we did for $2$ before).
Thus, let $|A(R_3)|=1$.
But then by XWNW we must have that $f_{4x}(R_1',R_2',R_{-1,2})>0$ and hence $f_{4top(R_4)}(R_1',R_2',R_{-1,2})<1$.
But then agent 4 may report $R_4':top(R_4)$ and we obtain a contradiction to SP as by XWNW $f_{4top(R_4)}(R_1',R_2',R_3,R_4')=1$.

%Next, let $|S|=3$, say $S=\{1,2,3\}$. Suppose to the contrary that for some $i\in S$, say $i=1$, we have $f_{1x}(R)< \frac{1}{3}$ and, w.l.o.g., that $|A(R_2)|\geq |A(R_3)|$. Let $R_1'$ be such that $R_1'|O=R_2|O$ and $|A(R_1')|=2$ and $R_2'$ be such that $R_2'|O=R_2|O$ and $A(R_2')=\{x\}$. Then \[ \frac{1}{3}> f_{1x}(R)=f_{1x}(R_1',R_{-1})=f_{2x}(R_1',R_{-1})=f_{2x}(R_1',R_2',R_{-1,2})=f_{1x}(R_1',R_2',R_{-1,2}),\] where first and third equality follow from SP while the second and fourth equality follow from ETAE. Using the same argument, we may suppose without loss of generality that $A(R_1')=A(R_2')=\{x\}$ and $R_1'|O=R_2'|O=R_2|O$. But now we obtain $f_{3x}(R_1',R_2',R_{-1,2})\leq \frac{1}{3}$ as otherwise $3$ may report $R_1'$ and we obtain a contradiction to feasibility using SP and ETAE. If $|A(R_3)|>1$, then we may repeat the same argument as above and let agent 1 report $R_1''$ such that $R_1''|O=R_3'|O$. Thus, let $|A(R_3)|=1$. But then by XWNW we must have that $f_{4x}(R_1',R_2',R_{-1,2})>0$ and $f_{4top(R_4)}(R_1',R_2',R_{-1,2})<1$. But then agent 4 may report $R_4':top(R_4)$ and we obtain a contradiction to SP and XWNW (as $top(R_4)\neq x$).

Last, let $f$ satisfy ETAE and $|S|=2$, say $S=\{1,2\}$. Suppose to the contrary that $f_{1x}(R)< \frac{1}{2}$. Using the same argument as above we then may suppose
$f_{1x}(R)=f_{2x}(R)< \frac{1}{2}$ and $A(R_1)=A(R_2)=\{x\}$. By XWNW, then we must have $f_{3x}(R)>0$ or $f_{4x}(R)>0$, say $f_{3x}(R)>0$ and $f_{3top(R_3)}(R)<1$. Now if agent 3 reports $R_3':top(R_3)$, then we continue to obtain
$f_{3top(R_3)}(R_3',R_{-3})<1$. But then by XWNW and $top(R_3)\neq x$,
$f_{4top(R_3)}(R_3',R_{-3})>0$ and $f_{4top(R_4)}(R_3',R_{-3})<1$.
Now if agent 4 reports $R_4':top(R_4)$, then we continue obtain
$f_{4top(R_4)}(R_3',R_4',R_{-3,4})<1$, which is a contradiction to XWNW as $x\neq top(R_3) \neq top(R_4)\neq x$.

\underline{\emph{(c):}} Let $R$ be a profile with $top(R_i)\neq i$ for all $i\in N$ and where all agents in $S= N\backslash \{j\}$ (thus $|S|=3$) rank $x$ and $y$ acceptable at the top in the same order while for $j\notin S$, $top(R_j)=z$. Without loss of generality, let $S=\{1,2,3\}$ and $j=4$.
Now by (b) we obtain $f_{ix}(R)=\frac{1}{3}$ for all $i\in S$. If $f$ satisfies XEFF, we further know that $4$ obtains none of $y$ (as 4 ranks
$y$ below $z$ while all others rank $y$ above $z$). Then, $y$ is split among agents in $S$ and we can use the same argument as in (a) to show that it is split uniformly.

Next, let $f$ satisfy ETAE. As all agents in $S$ find $x$ and $y$ acceptable in the same order, ETAE implies $f_{1y}(R)=f_{2y}(R)=f_{3y}(R)$. Using the same argument as in (b) we may suppose without loss of generality that $R_1=R_2=R_3:xy$, i.e., suppose that they consider $z$ unacceptable.
If $f_{1y}(R)<\frac{1}{3}$, then by XWNW we must have
$f_{4y}(R)>0$ and $f_{4z}(R)<1$. But now agent 4 may report $R_4':z$ and we obtain a contradiction between SP and XWNW.\hfill$\square$\newline

Our second lemma considers a condition under which an agent receives an object with probability zero when ranking it third.

\begin{lemma}\label{lem.none.at.bottom}
Let $|N|=4$, $|O|=3$ and $f$ satisfy SP, ETE and XWNW. If $f$ also satisfies ETAE or XEFF, then for any profile $R$ where two agents have preferences $R_i,R_j:x...$ while another agent $k$ ranks $x$ third, behind $y$ and $z$, we have $f_{kx}(R)=0$.
\end{lemma}
\noindent\textbf{Proof.} For XEFF, this is immediate, as otherwise $k$ could trade with $i$ or $j$. So suppose $f$ satisfies ETAE. If the fourth agent, say $m$, ranks $x$ first, then the claim follows from Lemma \ref{lem.equal.split.at.top} (b). So consider $top(R_m)\neq x$ and w.l.o.g. $top(R_m)=y$. If $k$ ranks $z$ first, it also follows from Lemma \ref{lem.equal.split.at.top} (b). If $R_k=yzx$, consider $R_k':zyx$ -- then  $f_{kx}(R_k',R_{-k})=0$ and, by SP,   $f_{kx}(R)=f_{kx}(R_k',R_{-k})=0$.\hfill$\square$\newline

Our third lemma considers a condition under which an object is split uniformly when all agents rank it second.

\begin{lemma}\label{lem.equal.split.at.second*}
Let $|N|=4$, $|O|=3$ and $f$ satisfy SP, ETE and XWNW. If $f$ also satisfies ETAE or XEFF, then for any profile $R$ where two agents have preferences $R_i,R_j:xyz$ while the other two agents have preferences $R_k,R_l:zy...$, each agent is assigned their second-most-preferred object $y$ with probability $\frac{1}{4}$.
\end{lemma}
\noindent\textbf{Proof.} %
%Let $R$ be a profile where all agents rank $y$ second, two agents, w.l.o.g. say $1$ and $2$, rank $x$ first, and the other two agents, $3$ and $4$, rank $z$ first. Suppose $f_{1y}(R)<\frac{1}{4}$. By SP, this holds regardless of whether $1$ ranks $z$ as acceptable third choice so w.l.o.g. assume he does. By ETAE, $f_{2y}(R)=f_{1y}(R)<\frac{1}{4}$ and we may likewise assume that $2$ ranks $z$ third. Moreover, by Lemma \ref{lem.none.at.bottom}, $f_{2z}(R)=f_{1z}(R)=0$ and hence, by ETAE, $f_{3z}(R)=f_{4z}(R)=\frac{1}{2}$ (as $z$ is assigned with probability 1 by XWNW and SP). Given that $f_{2y}(R)=f_{1y}(R)<\frac{1}{4}$ we also have $f_{3y}(R)=f_{4y}(R)>\frac{1}{4}$. By SP we may suppose that $3$ ranks $x$ third and by ETAE this also holds for $4$. Likewise, we may assume that $4$ ranks $x$. But then we have a preference profile where all agents consider all objects acceptable and where $3$ receives more than $\frac{3}{4}$. By SP, this would still hold if $3$'s preferences were $R_3'=R_1=R_2$. By ETAE, both $1$ and $2$ would then also receive more than $\frac{3}{4}$ in total. Thus $4$ must receive strictly less than $\frac{3}{4}$. By SP, this still holds as $4$, too, aligns her preferences with those of $1$ and $2$, and we have $R_4'=R_3'=R_1=R_2$. But then by ETAE, all agents have less than $\frac{3}{4}$ in total probability shares, violating XWNW.\hfill$\square$\newline
Let $R$ be a profile where all agents rank $y$ second, two agents, w.l.o.g. say $1$ and $2$, rank $x$ first, and the other two agents, $3$ and $4$, rank $z$ first. Moreover, suppose $1$ and $2$ rank all objects acceptable.

First, note that $1$ and $2$ receive none of $z$ by Lemma \ref{lem.none.at.bottom}. Then $z$ is split uniformly between $3$ and $4$ (otherwise we arrive at a contradiction by the same argument as in Lemma \ref{lem.equal.split.at.top} (a)).

Now, towards a contradiction, suppose that $f_{3y}(R)> \frac{1}{4}$, hence $3$ receives more than $\frac{3}{4}$ in total. By SP, $3$ still receives more than $\frac{3}{4}$ if she reports $R_3':xyz$. By ETE, $1$ and $2$ also receive strictly more of $\frac{3}{4}$ in total. By feasibility, $4$ must receive strictly less. Now, let $4$ report $R_4':zyx$ (possible $R_4'=R_4$). By SP, $4$ still receives strictly less of $z$ and $y$ and receives none of $x$ by Lemma \ref{lem.none.at.bottom} -- thus strictly less than $\frac{3}{4}$ in total. Finally, let $4$ report $R_4''=R_1=R_2=R_3'$. Then, by SP and ETE, all receive strictly less than $\frac{3}{4}$, contradicting XWNW. 

For a second contradiction, suppose that $f_{3y}(R)< \frac{1}{4}$, hence $3$ receives less than $\frac{3}{4}$ in total. Now, let $3$ report $R_3':zyx$ (possibly $R_3'=R_3$). By SP, $3$ still receives strictly less of $z$ and $y$ and receives none of $x$ by Lemma \ref{lem.none.at.bottom} -- thus strictly less than $\frac{3}{4}$ in total. As she reports $R_3'':xyz$, SP and ETE imply that $1$, $2$, and $3$ all receive strictly less than $\frac{3}{4}$ in total. By XWNW, $4$ must receive strictly more. Now, let $4$ report $R_4':xyz$, i.e., $R_4'=R_1=R_2=R_3''$. Then, by SP and ETE, all receive strictly more than $\frac{3}{4}$, contradicting feasibility. 

If $f_{4y}(R)\neq \frac{1}{4}$, we derive a contradiction in the same way as for $3$ above. 

If instead $f_{1y}(R)=f_{2y}(R)< \frac{1}{4}$ (where the equality follows from ETE), then the fact that $f_{3y}(R)=f_{4y}(R)= \frac{1}{4}$ implies that $y$ is assigned with probability less than $1$, violating XWNW. Similarly $f_{1y}(R)=f_{2y}(R)> \frac{1}{4}$ violates feasibility.\hfill$\square$\newline

\noindent
\textbf{Proof of Theorem \ref{the1}.}
Let $|N|\geq 4$, $|O|\geq 3$ and $f$ satisfy SP, ETE and XWNW. Moreover suppose $f$ also satisfies ETAE or XEFF.

By XWNW it suffices to show Theorem \ref{the1} for $N=\{1,2,3,4\}$ and $O=\{a,b,c\}$;  for  $N'\supseteq N$ and $O'\supseteq O$ the same arguments apply as long as all agents in $N$ consider only objects in $O$ acceptable and all objects in $O$ are considered acceptable only by agents in $N$.

Towards a contradiction, suppose that $f$ is at-most-$q$-agent-object-wasteful with $q<\frac{1}{6}$.

Consider profile $G$ in (\ref{RSD1}) which for convenience we repeat again below. Let $U=(bac,G_{-1})$. Then by Lemma \ref{lem.equal.split.at.top} (a), $f_{ib}(U)=\frac{1}{4}$ for all $i\in N$, and by XWNW $a$ is assigned with probability 1. Moreover, by ETE $f_{3a}(U)=f_{4a}(U)$, and by ETE and  SP $f_{1a}(U)=f_{3a}(U)=f_{4a}(U)=\frac{1}{3}$ (if $f_{1a}(U)\neq f_{3a}(U)=f_{4a}(U)$, letting $1$ declare $c$ unacceptable would yield a contradiction).

Let $U'=(abc,G_{-1})$. Then $f_{1b}(U')=0$ by Lemma \ref{lem.equal.split.at.top} (b). Thus by SP, $f_{1a}(U')=\frac{1}{3}+\frac{1}{4}= \frac{14}{24}$. The remaining entries in $f(U')$ follow from ETE and XWNW.
\[
\begin{array}{cccc}
G & a & b & c \\ \hline
ac & \frac{14}{24} & 0 &  \\
bc & 0 & \frac{1}{3} &  \\
ba & \frac{5}{24} & \frac{1}{3} & 0 \\
ba & \frac{5}{24} & \frac{1}{3} &  0
\end{array}\qquad
\begin{array}{cccc}
U & a & b & c \\ \hline
bac & \frac{1}{3} & \frac{1}{4} &  \\
bc & 0 & \frac{1}{4} &  \\
ba & \frac{1}{3} & \frac{1}{4} & 0 \\
ba & \frac{1}{3} & \frac{1}{4} &  0
\end{array}\qquad
\begin{array}{cccc}
U' & a & b & c \\ \hline
abc & \frac{14}{24} & 0 &  \\
bc & 0 & \frac{1}{3} &  \\
ba & \frac{5}{24} & \frac{1}{3} & 0 \\
ba & \frac{5}{24} & \frac{1}{3} &  0
\end{array}
\]
Thus, by SP, $f_{1a}(G)=\frac{14}{24}$, and by XWNW, $f_{1b}(G)=0$. As for $U'$, the remaining assignment probabilities for $f_a(G)$ follow from ETE and XWNW.
As $f$ is at-most-$q$-agent-object-wasteful with $q<\frac{1}{6}$, we must have $f_{1c}(G)>\frac{6}{24}$ or $f_{2c}(G)>\frac{14}{24}$.
We show $f_{2c}(G)=\frac{14}{24}$, which implies $f_{1c}(G)>\frac{6}{24}$.\medskip

For the profile $R$ below we show that
\[\begin{array}{cccc}
R & a & b & c \\ \hline
ac &  \frac{16}{24} & 0 & \frac{2}{24} \\
cab & 0 & 0 & \frac{22}{24}\\
ba & \frac{4}{24} & \frac{1}{2} & 0 \\
ba & \frac{4}{24} & \frac{1}{2} &  0
\end{array}.
\]
Note that by SP and $R_{-2}=G_{-2}$ we then obtain $f_{2c}(G)=\frac{14}{24}$. Furthermore, we obtain $f_{3b}(R)=f_{4b}(R)=\frac{1}{2}$ by Lemma \ref{lem.equal.split.at.top} (b). Hence, $f_{1b}(R)=f_{2b}(R)=0$.

\medskip
\noindent
\textit{Claim 1: $f_{1c}(R)=\frac{2}{24}$ and $f_{2c}(R)=\frac{22}{24}$.}\medskip

Note that $c$ is assigned with probability 1 as otherwise if $c$ is unassigned for some weakly non-wasteful assignment in the support of $f(R)$, then agents 1 and 2 cannot be unassigned, agent 1 must obtain $a$ and agent 2 must obtain $b$, which is then a contradiction to the fact that $f_{2b}(R)=0$.

Thus, it suffices to show that $f_{1c}(R)=\frac{2}{24}$, which we do in two steps: \emph{(1)} $f_{1a}(R)+f_{1c}(R)=\frac{3}{4}$ and \emph{(2)} $f_{1a}(R)=\frac{2}{3}$.

\noindent
Step \emph{(1)}:
The first main difficulty is to show $f_{1c}(Q)=\frac{5}{6}$, $f_{1a}(Q)=0$ and $f_{1b}(Q)=0$ in profile $Q$ below.
\[\begin{array}{cccc}
Q & a & b & c \\ \hline
cab &  &  &  \\
bac &  &  & \\
bac &  &  & \\
ba &  &  &
\end{array}
\qquad \begin{array}{cccc}
Q^0 & a & b & c \\ \hline
bac & \frac{1}{4} & \frac{1}{4} & \frac{1}{3} \\
bac & \frac{1}{4} & \frac{1}{4} & \frac{1}{3}\\
bac & \frac{1}{4} & \frac{1}{4} & \frac{1}{3}\\
ba & \frac{1}{4} & \frac{1}{4} &  0
\end{array}\]
If all have preferences $bac$, then by ETE $a$ and $b$ are split equally. Thus, for $Q^0$  SP implies that $f_{4a}(Q^0)=f_{4b}(Q^0)=\frac{1}{4}$. The remaining entries then follow from ETE and XWNW. By $Q^0_{-1}=Q_{-1}$ and SP we obtain $\sum_{o\in O}f_{1o}(Q)=\frac{5}{6}$.

Consider the following sequence of profiles:
\[\begin{array}{cccc}
Q & a & b & c \\ \hline
cab & 0 & 0 & \frac{5}{6} \\
bac & \frac{1}{3} & \frac{1}{3} & \frac{1}{12}\\
bac & \frac{1}{3} & \frac{1}{3} & \frac{1}{12}\\
ba & \frac{1}{3} & \frac{1}{3} &  0
\end{array}
\qquad
\begin{array}{cccc}
Q' & a & b & c \\ \hline
cab & \frac{1}{4} & 0 & \frac{1}{2} \\
cab & \frac{1}{4} & 0 & \frac{1}{2}\\
bac & \frac{1}{4} & \frac{1}{2} & 0\\
ba & \frac{1}{4} & \frac{1}{2} &  0
\end{array}\qquad
\begin{array}{cccc}
Q'' & a & b & c \\ \hline
cab & \frac{1}{4} & 0 & \frac{1}{2} \\
cab & \frac{1}{4} & 0 & \frac{1}{2}\\
ba & \frac{1}{4} & \frac{1}{2} & 0\\
ba & \frac{1}{4} & \frac{1}{2} &  0
\end{array}
\qquad
\begin{array}{cccc}
Q''' & a & b & c \\ \hline
ca & \frac{1}{4} & 0 & \frac{1}{2} \\
cab &  &  & \\
ba &  &   &\\
ba &  &  &
\end{array}
\]

\[\begin{array}{cccc}
\hat{Q} & a & b & c \\ \hline
cab & 0 & 0 & \frac{3}{4} \\
bac & \frac{1}{3} & \frac{1}{3} & \frac{1}{12}\\
bac & \frac{1}{3} & \frac{1}{3} & \frac{1}{12}\\
bac & \frac{1}{3} & \frac{1}{3} & \frac{1}{12}
\end{array}
\qquad
\begin{array}{cccc}
\hat{Q}' & a & b & c \\ \hline
cab & \frac{1}{4} & 0 & \frac{1}{2} \\
cab & \frac{1}{4} & 0 & \frac{1}{2}\\
bac & \frac{1}{4} & \frac{1}{2} & 0\\
bac & \frac{1}{4} & \frac{1}{2} &  0
\end{array}
\qquad
%\begin{array}{cccc}
%Q'' & a & b & c \\ \hline
%cab & \frac{1}{4} & 0 & \frac{1}{2} \\
%cab & \frac{1}{4} & 0 & \frac{1}{2}\\
%ba & \frac{1}{4} & \frac{1}{2} & 0\\
%ba & \frac{1}{4} & \frac{1}{2} &  0
%\end{array}
%\qquad
%\begin{array}{cccc}
%Q''' & a & b & c \\ \hline
%ca & \frac{1}{4} & 0 & \frac{1}{2} \\
%cab &  &  & \\
%ba &  &   &\\
%ba &  &  &
%\end{array}
\]

First, the entries for $a$ and $b$ at $\hat{Q}$ follow by Lemma \ref{lem.equal.split.at.top} (c) as $a$ and $b$ are split equally among agents $2$, $3$ and $4$. Moreover, by SP we know that the total assignment probabilities for $1$ are $\frac{3}{4}$ (consider a change in $1$'s preferences to $bac$ and invoke ETE and XWNW). This pins down $f(\hat{Q})$.

For $Q$, we obtain the assignment probabilities for $b$ and $a$ by Lemma \ref{lem.equal.split.at.top} (b). Thus $f_{1c}(Q)=\sum_{o\in O}f_{1o}(Q)=\frac{5}{6}$. As three agents rank $c$ acceptable, object $c$ is fully assigned (by XWNW) and we obtain $f_{2c}(Q)=f_{3c}(Q)=\frac{1}{12}$ by ETE.

For $\hat{Q}'$ we know that the total assignment probabilities of $2$ sum to $\frac{3}{4}$ (by $\hat{Q}_{-2}=\hat{Q}_{-2}'$ and SP). By ETE the same is true for $1$. But then the same must hold for $3$ and $4$ since all objects are assigned with probability 1 by XWNW.
By Lemma \ref{lem.equal.split.at.second*}, we have $f_{1a}(\hat{Q}')=f_{2a}(\hat{Q}')=f_{3a}(\hat{Q}')=f_{4a}(\hat{Q}')=\frac{1}{4}$.
By Lemma \ref{lem.none.at.bottom}, $f_{1b}(\hat{Q}')=f_{2b}(\hat{Q}')=0=f_{3c}(\hat{Q}')=f_{4c}(\hat{Q}')$. Now all remaining entries
of $f(\hat{Q}')$ follow from ETE and XWNW.
%For (ii), XEFF implies $f_{1b}(\hat{Q}')=f_{2b}(\hat{Q}')=0=f_{3c}(\hat{Q}')=f_{4c}(\hat{Q}')$ and then this determines $f(\hat{Q}')$.
%Next assume that $f_{3b}(\hat{Q}')<\frac{1}{2}$. Then consider a change in $3$'s preferences to $ba$. By SP, $3$ would still receive less than $\frac{1}{2}$ and by ETAE the same would hold for $4$. Moreover, the same holds true when $4$ reports $b$ as the only acceptable object (while not changing the order of the objects) and by ETAE the same holds for 3. But then also 3 may just report $b$ as the only acceptable object (while not changing the order of the objects), and then $1$ (and $2$) would have to receive $b$ with positive probability. Thus, $a$ or $c$ is assigned with total probability less than $1$. Let $1$ drop $b$ and then by SP, 1 is unassigned with positive probability, which then contradicts XWNW as $c$ and $a$ would have to be assigned whenever 1 is unassigned but only agent 2 ranks $c$ and $a$ acceptable.
%Thus, we conclude that $f_{3b}(\hat{Q}')=\frac{1}{2}=f_{4b}(\hat{Q}')$ and similarly, $f_{1c}(\hat{Q}')=\frac{1}{2}=f_{2c}(\hat{Q}')$. This determines also $f(\hat{Q}')$.

Moving to $Q'$, $Q_{-2}'=Q_{-2}$ and SP imply that $2$'s probability shares remain at the sum total of $\frac{3}{4}$. By ETE, the same holds for $1$. %By $\hat{Q}_{-4}'=Q_{-4}'$, SP and ETAE, we obtain $f_{3a}(Q')=f_{4a}(Q')=\frac{1}{4}$ and $f_{3b}(Q')=f_{4b}(Q')=\frac{1}{2}$.
By Lemma \ref{lem.equal.split.at.second*} $f_{3a}(Q')=f_{4a}(Q')=\frac{1}{4}$ and by Lemma \ref{lem.none.at.bottom} we know that $b$ is split among $3$ and $4$. Moreover, by SP and ETE the split must be uniform, $f_{3b}(Q')=f_{4b}(Q')=\frac{1}{2}$ (otherwise, letting $3$ declare $c$ unacceptable yields a contradiction).
By XWNW, ETE, and the fact that all three objects are fully assigned this determines $f(Q')$.
Moving to $Q''$, by $Q_{-3}'=Q_{-3}''$ and SP, we obtain $f_{3a}(Q'')=\frac{1}{4}$ and $f_{3b}(Q'')=\frac{1}{2}$. By ETE, 4 receives the same shares as 3, and by ETE and XWNW we obtain $f(Q'')$. Finally, moving to $Q'''$, by $Q_{-1}'''=Q_{-1}''$ and SP, $1$ still receives $\frac{3}{4}$ for objects $a$ and $c$. As $R_1:ac$ and $R_{-1}=Q_{-1}'''$, SP now completes Step \emph{(1)}.

\medskip
\noindent
Step \emph{(2)}:
\[
\begin{array}{cccc}
R' & a & b & c \\ \hline
ba & \frac{1}{3} & \frac{1}{3} & 0 \\
cab & 0 & 0 & 1\\
ba & \frac{1}{3} & \frac{1}{3} & 0 \\
ba & \frac{1}{3} & \frac{1}{3} &  0
\end{array}\qquad
\begin{array}{cccc}
R'' & a & b & c \\ \hline
ab &  \frac{2}{3} & 0 & 0 \\
cab &  &  & \\
ba &  &  & \\
ba &  &  &
\end{array}
\qquad
\]
For profile $R'$, we obtain $f_{2a}(R')=0$ and $f_{2b}(R')=0$ by Lemma \ref{lem.equal.split.at.top} (c). But then by XWNW we obtain $f_{2c}(R')=1$ (as 2 is the only agent ranking object $c$ acceptable). For profile $R''$ we have $f_{1b}(R'')=0$ by Lemma \ref{lem.equal.split.at.top} (b) (as $f_{3b}(R'')=f_{4b}(R'')=\frac{1}{2}$).
Thus, by  SP, we obtain $f_{1a}(R'')=\frac{2}{3}$. But then, moving from $R''$ to $R$, we have by SP, $f_{1a}(R)=\frac{2}{3}=\frac{16}{24}$. This completes Step \emph{(2)}.\medskip

\medskip
\noindent
\textit{Claim 2: $f_{2c}(G)=\frac{14}{24}$.}

Since $f_{2a}(G)=0$ and $f_{2b}(G)=\frac{1}{3}$, the claim is equivalent to $\sum_{o\in O} f_{2o}(G)=\frac{22}{24}$. By SP, it suffices to show $\sum_{o\in O} f_{2o}(R)=\frac{22}{24}$ in the profile $R$ above.

We know already $f_{2b}(R)=0$ (by Lemma \ref{lem.equal.split.at.top} (b)) and $f_{2c}(R)=\frac{22}{24}$ (by Claim 1). Suppose that $f_{2a}(R)>0$.
Then consider the following sequence of profiles:
\[\begin{array}{cccc}
\hat{R} & a & b & c \\ \hline
ac & >\frac{5}{12} & 0 & \frac{1}{2} \\
ac & >\frac{5}{12} & 0 & \frac{1}{2}\\
ba & <\frac{1}{12} & \frac{1}{2} & 0 \\
ba & <\frac{1}{12} & \frac{1}{2} &  0
\end{array}
\qquad
\begin{array}{cccc}
\hat{R}' & a & b & c \\ \hline
ac & \frac{1}{3} & 0 & \frac{1}{2} \\
ac & \frac{1}{3} & 0 & \frac{1}{2}\\
ab & \frac{1}{3} & \frac{1}{4} & 0 \\
ba & 0 & \frac{3}{4} &  0
\end{array}
\qquad
\begin{array}{cccc}
\hat{R}'' & a & b & c \\ \hline
ac & \frac{1}{4} & 0 & \frac{1}{2} \\
ac & \frac{1}{4} & 0 & \frac{1}{2}\\
ab & \frac{1}{4} & \frac{1}{2} & 0 \\
ab & \frac{1}{4} & \frac{1}{2} &  0
\end{array}
\]
For profile $\hat{R}$ from ETE and XWNW we obtain the random assignments for objects $c$ and $b$. By $\hat{R}_{-2}=R_{-2}$, SP and ETE we obtain $f_{1a}(\hat{R})=f_{2a}(\hat{R})>\frac{5}{12}$ (as $f_{2a}(R)+f_{2c}(R)>\frac{22}{24}=\frac{1}{2}+\frac{5}{12}$).
Thus, by ETE and XWNW, $f_{3a}(\hat{R})=f_{4a}(\hat{R})<\frac{1}{12}$.
For profile $\hat{R}''$, $f(\hat{R}'')$ is obtained using ETE, XWNW, and Lemma \ref{lem.equal.split.at.top} (a) (which implies that $a$ is shared equally).
For profile $\hat{R}'$ we must have $f_{4a}(\hat{R}')=0$ as by Lemma \ref{lem.equal.split.at.top} (b) we have
$f_{1a}(\hat{R}')=f_{2a}(\hat{R}')=f_{3a}(\hat{R}')=\frac{1}{3}$. Thus, by $\hat{R}_{-4}'=\hat{R}_{-4}''$ and SP we obtain $f_{4b}(\hat{R}')=\frac{3}{4}$. As $b$ is fully assigned under $\hat{R}'$, this now implies $f_{3b}(\hat{R}')=\frac{1}{4}$.
But then we have
\[ f_{3a}(\hat{R}')+f_{3b}(\hat{R}')=\frac{1}{3}+\frac{1}{4}=\frac{7}{12}> f_{3a}(\hat{R})+f_{3b}(\hat{R}),\]
which is a contradiction to SP as $\hat{R}_{-3}=\hat{R}_{-3}'$ and both $\hat{R}_3$ and $\hat{R}_3'$ rank $a$ and $b$ as the only acceptable objects. This establishes $f_{2a}(R)=0$ and hence completes the proof of the Claim 2.\medskip

Thus, by Claim 2, $f_{2c}(G)= \frac{14}{24}$, and we have $f_{1c}(G)>\frac{6}{24}$ as $f$ is at-most-$q$-agent-object-wasteful with $q<\frac{1}{6}$. To lead this remaining case to a contradiction, consider the following sequence of profiles:
\[\begin{array}{cccc}
Q & a & b & c \\ \hline
ac &  \frac{14}{24} & 0 & >\frac{6}{24} \\
bca & 0 & \frac{1}{3} & \frac{14}{24}\\
ba &  & \frac{1}{3} & 0 \\
ba &  & \frac{1}{3} &  0
\end{array}
\qquad
\begin{array}{cccc}
Q' & a & b & c \\ \hline
acb & \frac{14}{24} & 0 & >\frac{6}{24} \\
bca &  &  & \frac{1}{2}\\
ba & \frac{1}{4} &  & 0 \\
ba & \frac{1}{4} &  &  0
\end{array}\qquad
\begin{array}{cccc}
Q'' & a & b & c \\ \hline
bca & \alpha>\frac{2}{24} & \frac{1}{4} & \frac{1}{2}  \\
bca & \alpha>\frac{2}{24} & \frac{1}{4} & \frac{1}{2}\\
ba & \frac{1}{2}-\alpha <\frac{10}{24} & \frac{1}{4} & 0 \\
ba & \frac{1}{2}-\alpha <\frac{10}{24} & \frac{1}{4} &  0
\end{array}
\]
Note that for profile $Q$ we obtain by $R_{-2}=Q_{-2}=G_{-2}$, SP and $f_{2a}(R)=0$ that $f_{2a}(Q)=0$. Then $f_{2c}(Q)=\frac{14}{24}$ follows from $f_{2c}(R)=\frac{22}{24}$ and $f_{2b}(Q)=\frac{1}{3}$ (as by Lemma \ref{lem.equal.split.at.top} (b) we have $f_{2b}(Q)=f_{3b}(Q)=f_{4b}(Q)=\frac{1}{3}$, and similarly for profile $Q'$ in obtaining $f_{1b}(Q')=0$).
Using the same argument as for determining $f_{1a}(G)=\frac{14}{24}$ we obtain $f_{1a}(Q)=\frac{14}{24}$. Hence, as $f$ is at-most-$q$-agent-object-wasteful with $q<\frac{1}{6}$, we must have $f_{1c}(Q)> \frac{6}{24}$.
For profile $Q'$, by $f_{1b}(Q')=0$, $Q_{-1}'=Q_{-1}$ and SP, $f_{1a}(Q')=\frac{14}{24}$ and $f_{1c}(Q')>\frac{6}{24}$.
For profile $Q''$, by Lemma \ref{lem.equal.split.at.top} (a), $f_{1b}(Q'')=f_{2b}(Q'')=f_{3b}(Q'')=f_{4b}(Q'')=\frac{1}{4}$. Moreover, $f_{1c}(Q'')=f_{2c}(Q'')=\frac{1}{2}$ by letting agents 1 and 2 drop $a$ and invoking ETAE, or by XEFF and ETE.
Thus, by $Q_{-1}'=Q_{-1}''$, SP and ETE, $f_{2a}(Q'')=f_{1a}(Q'')=\alpha>\frac{2}{24}$. As all agents rank $a$ acceptable, by XWNW object $a$ is assigned with probability 1 and we obtain by ETE,
$f_{3a}(Q'')=f_{4a}(Q'')=\frac{1}{2}-\alpha<\frac{10}{24}$. Note that $f_{3a}(Q'')+f_{3b}(Q'')<\frac{1}{4}+\frac{10}{24}=\frac{2}{3}$.

Next consider the following two profiles.
\[\begin{array}{cccc}
\hat{Q} & a & b & c \\ \hline
bca &   & \frac{1}{3} & \frac{1}{2} \\
bca &  & \frac{1}{3} & \frac{1}{2} \\
ab & \frac{3}{4}-\alpha & 0 & 0 \\
ba &  & \frac{1}{3} &  0
\end{array}
\qquad
\begin{array}{cccc}
\hat{Q}' & a & b & c \\ \hline
bca &  \frac{1}{24}+\frac{\alpha}{2} & \frac{1}{4} & \frac{1}{2}-\frac{\alpha}{2} \\
bca & \frac{1}{24}+\frac{\alpha}{2} & \frac{1}{4} & \frac{1}{2}-\frac{\alpha}{2}\\
bac & \frac{1}{2}-\alpha  & \frac{1}{4}  &  \alpha \\
ba &   \frac{5}{12}  & \frac{1}{4} &  0
\end{array}
\]
For profile $\hat{Q}$, we obtain $f_{b}(\hat{Q})$ by Lemma \ref{lem.equal.split.at.top} (b). Hence, by $f_{3b}(\hat{Q})=0$, $\hat{Q}_{-3}=Q_{-3}''$, SP, and $\alpha>\frac{2}{24}$, we have $f_{3a}(\hat{Q})=\frac{3}{4}-\alpha<\frac{2}{3}$.
Next $f_{1b}(\hat{Q})=f_{2b}(\hat{Q})=f_{4b}(\hat{Q})=\frac{1}{3}$ by Lemma \ref{lem.equal.split.at.top} (b).
Moreover, $f_{1c}(\hat{Q})=f_{2c}(\hat{Q})=\frac{1}{2}$ -- either by XEFF and ETE or by letting $1$ and $2$ drop $a$ and invoking ETAE and XWNW.
For profile $\hat{Q}'$, by Lemma \ref{lem.equal.split.at.top} (a), $f_{1b}(\hat{Q}')=f_{2b}(\hat{Q}')=f_{3b}(\hat{Q}')=f_{4b}(\hat{Q}')=\frac{1}{4}$.
Hence, by $\hat{Q}_{-3}=\hat{Q}_{-3}'$ and SP, we have $f_{3a}(\hat{Q}')=\frac{1}{2}-\alpha$ and by SP (using $\hat{Q}_{-3}'=U_{-3}'$ from below), $f_{3c}(\hat{Q}')=\alpha$. By ETE and XWNW we obtain $f_{1c}(\hat{Q}')=f_{2c}(\hat{Q}')={\frac{1}{2}} - {\frac{\alpha}{2}}$ as $c$ is fully assigned.

To determine the remaining entries for $\hat{Q}'$, consider
\[
\begin{array}{cccc}
U & a & b & c \\ \hline
bca & \frac{1}{6} & \frac{1}{4} & \frac{1}{3} \\
bca & \frac{1}{6} & \frac{1}{4} & \frac{1}{3}\\
bca & \frac{1}{6}  & \frac{1}{4}  &  \frac{1}{3}\\
bac &  \frac{1}{2}  & \frac{1}{4} & 0
\end{array}
\qquad
\begin{array}{cccc}
U' & a & b & c \\ \hline
bca & \frac{1}{6} & \frac{1}{4} & \frac{1}{3} \\
bca & \frac{1}{6} & \frac{1}{4} & \frac{1}{3}\\
bca & \frac{1}{6}  & \frac{1}{4}  &  \frac{1}{3}\\
ba &  \frac{1}{2}  & \frac{1}{4} & 0
\end{array}
\qquad
\begin{array}{cccc}
\hat{Q}'' & a & b & c \\ \hline
bca & 0 & \frac{1}{4} & \frac{7}{12} \\
bac & \frac{1}{3} & \frac{1}{4} & \frac{5}{24}\\
bac & \frac{1}{3}  & \frac{1}{4}  &  \frac{5}{24}\\
ba &  \frac{1}{3}  & \frac{1}{4} &
\end{array}
\qquad
\begin{array}{cccc}
\hat{Q}''' & a & b & c \\ \hline
bac & \frac{1}{4} & \frac{1}{4} & \frac{1}{3} \\
bac & \frac{1}{4} & \frac{1}{4} & \frac{1}{3}\\
bac & \frac{1}{4}  & \frac{1}{4}  &  \frac{1}{3}\\
ba &  \frac{1}{4}  & \frac{1}{4} & 0
\end{array}
\]
For profile $U$, by (a) of Lemma \ref{lem.equal.split.at.top} we have $f_{1b}(U)=f_{2b}(U)=f_{3b}(U)=f_{4b}(U)=\frac{1}{4}$.
We obtain
$f_{1c}(U)=f_{2c}(U)=f_{3c}(U)=\frac{1}{3}$ by invoking ETAE or XEFF and SP as in the proof of (b) of Lemma \ref{lem.equal.split.at.top} (where agents 1, 2 and 3 drop $a$). But then by considering $(bca,U_{-4})$ we obtain by SP that $f_{4a}(U)=\frac{1}{2}$. As all objects are fully assigned by XWNW, now we obtain $f_{1a}(U)=f_{2a}(U)=f_{3a}(U)=\frac{1}{6}$.
In moving to $U'$ we obtain by $U_{-4}=U_{-4}'$ and SP, $f_{4a}(U')=\frac{1}{2}$ and $f_{4b}(U')=\frac{1}{4}$. Now the remaining entries in $f(U')$ follow from ETE, XWNW and the fact that all three objects are fully assigned.

Then $f_{4a}(\hat{Q}''')=\frac{1}{4}$ (consider $4$ declaring $c$ acceptable and invoke SP, ETE, and XWNW), while the remaining entries follow from ETE and XWNW. At $\hat{Q}''$ we have $f_{1a}(\hat{Q}'')=0$ if $f$ satisfies XEFF; if $f$ satisfies ETAE and, towards a contradiction, we had $f_{1a}(\hat{Q}'')>0$, then $f_{2a}(\hat{Q}'')=f_{3a}(\hat{Q}'')=f_{4a}(\hat{Q}'')<\frac{1}{3}$. But then this remains the case even as $2$ and $3$ drop $c$ -- contradicting XWNW as $1$ is unassigned with positive probability and whenever that is the case, $c$ is unassigned as well. The remaining entries at $\hat{Q}''$ follow by SP and ETE.

Now comparing $\hat{Q}''$ and $\hat{Q}'$ SP yields $f_{2a}(\hat{Q}')=f_{2a}(\hat{Q}'')+f_{2c}(\hat{Q}'')-f_{2c}(\hat{Q}')=\frac{1}{24}+\frac{\alpha}{2}$ and, by ETE, $f_{1a}(\hat{Q}')=\frac{1}{24}+\frac{\alpha}{2}$. As residual, since $a$ is fully assigned by XWNW, we get $f_{4a}(\hat{Q}')=\frac{5}{12}$.

If $f$ satisfies ETAE, comparing $3$ and $4$ yields $\frac{1}{2}-\alpha=f_{3a}(\hat{Q}')=f_{4a}(\hat{Q}')=\frac{5}{12}$, thus contradicting $\alpha>\frac{2}{24}$.

It remains to derive a contradiction for the case where $f$ satisfies XEFF (even as it may fail ETAE). For this we show that $f_{1a}(\hat{Q}')(=f_{2a}(\hat{Q}'))=\frac{1}{12}$, which establishes $\alpha=\frac{1}{12}$, the desired contradiction to $\alpha> \frac{2}{24}$.

To do so, first consider the following preference profiles, arrived at by varying $1$'s preferences relative to $\hat{Q}'$.
\[\begin{array}{cccc}
W & a & b & c \\ \hline
bc &  0 & \frac{1}{4} & \frac{1}{2}-\frac{\alpha}{2} \\
bca &  & \frac{1}{4} &  \\
bac & & \frac{1}{4} &  \\
ba &  & \frac{1}{4} &  0
\end{array}
\qquad
\begin{array}{cccc}
W' & a & b & c \\ \hline
cb &  0 & 0 & \frac{3}{4}-\frac{\alpha}{2} \\
bca &  & \frac{1}{3} & \\
bac &   & \frac{1}{3}  &   \\
ba &    & \frac{1}{3} &  0
\end{array}
\]
The entries for $b$ follow from Lemma \ref{lem.equal.split.at.top}, the remaining entries for $1$ by SP. Now at $W'$ that leaves $1$ unassigned with probability $\frac{1}{4}+\frac{\alpha}{2}$.

Crucially, we do know by XEFF that at $W'$ all objects are assigned with probability $1$ -- for $b$ and $c$ this follows since they are top ranked, for $a$ this is the case as three agents rank it third or higher. Hence, if we determine the probabilities of $2$, $3$ and $4$ to be unassigned, that allows us to determine the probability with which $1$ is unassigned (as the total unassignment probabilities should sum to 1), and hence to determine $\alpha$.

We proceed in three steps. To determine the probability with which 2 is unassigned, we first determine $f_{2c}(W')+f_{2a}(W')$. Then, to determine the probability that $3$ is unassigned we consider $f_{3c}(W')+f_{3a}(W')$. Last, we determine the probability that $4$ is unassigned at $W'$.

\underline{Step 1}: $f_{22}(W')=\frac{1}{4}$. For this consider
\[
\begin{array}{cccc}
W'' & a & b & c \\ \hline
cb &  0 & 0 & \frac{5}{6} \\
bac &  \frac{1}{3} & \frac{1}{3} & \frac{1}{12}\\
bac &   \frac{1}{3} & \frac{1}{3}  &  \frac{1}{12} \\
ba &    \frac{1}{3} & \frac{1}{3} &  0
\end{array}
\qquad
\begin{array}{cccc}
\tilde{W}''' & a & b & c \\ \hline
cb(a) &  0 & 0 & \frac{5}{6} \\
bac &  &  & \\
bac &   &   &   \\
bac &  \frac{1}{3}  & \frac{1}{3} &
\end{array}
\qquad
\qquad
\begin{array}{cccc}
W''' & a & b & c \\ \hline
bac &  \frac{1}{4} & \frac{1}{4} & \frac{1}{3} \\
bac &  &  & \\
bac &   &   &   \\
ba &  \frac{1}{4}  & \frac{1}{4} &  0
\end{array}
\qquad
\begin{array}{cccc}
W'''' & a & b & c \\ \hline
bac &  \frac{1}{4} & \frac{1}{4} & \frac{1}{3}   \\
bac &  \frac{1}{4} & \frac{1}{4} & \frac{1}{3} \\
bac & \frac{1}{4} & \frac{1}{4} & \frac{1}{3}    \\
bac & \frac{1}{4} & \frac{1}{4} & \frac{1}{3}
\end{array}.
\]
The entries for $4$ in $W''''$ follow by ETE and XEFF, for $4$ at $W'''$ by SP and XEFF. The entries for $1$ in $W'''$ then follow by ETE and XEFF. By XEFF and SP we get the entries for $1$ at $W''$. The remaining probabilities for $a$ and $b$ at $W''$  follow from ETE and SP (consider also moving from $W''''$ to $W''$ via $\tilde{W}'''$). The entries for $c$ at $W''$ follow as residuals by ETE. But then $f_{2c}(W')+f_{2a}(W')=f_{2c}(W'')+f_{2a}(W'')=\frac{5}{12}$. Hence the probability that $2$ is unassigned at $W'$ is $\frac{1}{4}$.

\underline{Step 2}: $f_{33}(W')=\frac{5}{24}$.  For this consider
\[
\begin{array}{cccc}
W'' & a & b & c \\ \hline
cb &  0 & 0 & \\
bca &  \frac{1}{4} & \frac{1}{3} & \\
bca &   \frac{1}{4} & \frac{1}{3}  &   \\
ba &    \frac{1}{2} & \frac{1}{3} &  0
\end{array}
\qquad
\begin{array}{cccc}
W''' & a & b & c \\ \hline
cb &  0 & 0 & \frac{1}{2} \\
bca &  &  &\\
bca &     & & \\
bca &   \frac{1}{3}  &  \frac{1}{3} &  \frac{1}{6}
\end{array}
\qquad
\begin{array}{cccc}
W'''' & a & b & c \\ \hline
cba & \frac{1}{4} & 0 & \frac{1}{2}  \\
bca &   &  &\\
bca &     & & \\
bca &   &  &
\end{array}.
\]
Solving from $W''''$ to $W''$ we get $f_{4a}(W'')=\frac{1}{2}$, and as residual by ETE  $f_{3a}(W'')=\frac{1}{4}$. By XEFF, $f_{1b}(W'')=0$; as a residual and by ETE we have $f_{2b}(W'')=f_{3b}(W'')=\frac{1}{3}$.

Next, consider
\[
\begin{array}{cccc}
W'' & a & b & c \\ \hline
cb &   &  &\frac{7}{12} \\
bca &  & & \frac{5}{24}  \\
bca &   & & \frac{5}{24}    \\
ba &    & &   0
\end{array}
\qquad
\begin{array}{cccc}
V & a & b & c \\ \hline
cba & \frac{1}{6}  & 0 & \frac{7}{12} \\
bca &  &   & \frac{5}{24}\\
bca &  &    & \frac{5}{24} \\
ba(c) &  &   &     0
\end{array}
\qquad
\begin{array}{cccc}
V' & a & b & c \\ \hline
bca &  \frac{1}{6} & \frac{1}{4}& \frac{1}{3} \\
bca & \frac{1}{6} &  \frac{1}{4} &\frac{1}{3}\\
bca & \frac{1}{6} &  \frac{1}{4}  & \frac{1}{3}\\
ba(c) & \frac{1}{2} & \frac{1}{4}  &     0
\end{array}
\]
For $V'$ and $V$ we consider two versions, one with  $c$ acceptable for $4$ and one without. In both cases, $b$ is split uniformly at $V'$ by Lemma \ref{lem.equal.split.at.top} (a) and $4$ gets none of $c$ by XEFF. The remaining probabilities at $V'$ follow by ETE and XEFF. Moving to $V$, $1$ gets none of $b$ by XEFF and hence $\frac{7}{12}$ of $c$ by SP. The remaining assignment probabilities for $c$ follow as residuals by ETE. Moreover, moving to $W''$, $1$ still gets $\frac{7}{12}$ of $c$ by SP and hence, by ETE, $2$ and $3$ still get $\frac{5}{24}$.

But then $f_{3c}(W')+f_{3a}(W')=f_{3c}(W'')+f_{3a}(W'')=\frac{5}{24}+\frac{1}{4}=\frac{11}{24}$. Hence, at $W'$, $3$ is unassigned with probability $\frac{5}{24}$.

\underline{Step 3}: $f_{44}(W')=\frac{1}{4}$.
For this
consider
\[
\begin{array}{cccc}
W'' & a & b & c \\ \hline
cb &   & 0 & \\
bca &   & \frac{1}{2} & \frac{1}{6} \\
bac & \frac{1}{6}  & \frac{1}{2} &     \\
ab &    & 0 &
\end{array}
\quad
\begin{array}{cccc}
W''' & a & b & c \\ \hline
cb &   &  & \\
bca &  &  &  \\
ba &  \frac{1}{6} & \frac{1}{2}&     \\
ab &    & &
\end{array}
\quad
\begin{array}{cccc}
W'''' & a & b & c \\ \hline
cb &   & 0 &  \\
bca &  & \frac{2}{3}&  \\
ab & \frac{1}{2}  & \frac{1}{6}&     \\
ab & \frac{1}{2}   & \frac{1}{6} &
\end{array}
\quad
\begin{array}{cccc}
W''''' & a & b & c \\ \hline
cb &   &  & \\
ba & 0 & \frac{2}{3} &  \\
ab &   & &     \\
ab &    & &
\end{array}
\quad
\begin{array}{cccc}
W'''''' & a & b & c \\ \hline
cb &   &  & 1\\
ab & \frac{1}{3} & \frac{1}{3}&  \\
ab &   & &     \\
ab &    & &
\end{array}
\]
At $W''''''$, $a$ and $b$ are split uniformly among three agents while $1$ is satiated with $c$ (XEFF and ETE). Entries for $2$ at $W'''''$ follow by XEFF and SP. By SP, $2$ also gets $\frac{2}{3}$ of $b$ at $W''''$, remaining entries for $b$ follow from XEFF (in case of $1$) and by ETE as residuals. For $a$, XEFF implies that it is split between $3$ and $4$ and uniformly so by ETE. At $W'''$, $3$ still is assigned with probability $\frac{2}{3}$ by SP and by Lemma 4 (b) she gets $b$ with probability $\frac{1}{2}$. By SP this give us two entries at $W''$ for agent $3$.

By an analogous argument we get the corresponding entries for agent $2$ at $W''$. As residual probabilities, we find that neither $1$ nor $4$ is assigned any share of $b$. In the next step we determine the probability with which agents 2 and 3 are assigned their respective third-most-preferred object. This then completes the assignment probabilities for $W''$.
\[
\begin{array}{cccc}
W'' & a & b & c \\ \hline
cb &   &  & \\
bca &   &  &  \\
bac &   &  &     \\
ab &    &  &
\end{array}
\quad
\begin{array}{cccc}
V & a & b & c \\ \hline
cb &   &  & \\
bac    & \frac{1}{6}& & \\
bac    & \frac{1}{6}  & &    \\
ab    & \frac{2}{3} & 0 &
\end{array}
\quad
\begin{array}{cccc}
V' & a & b & c \\ \hline
cb &   &  & \\
bac &   & &  \\
bac &   &  &     \\
abc   & \frac{2}{3} & 0 & \frac{1}{12}
\end{array}
\quad
\begin{array}{cccc}
V'' & a & b & c \\ \hline
cb &  0 & 0 & \frac{3}{4} \\
bac &   & &  \\
bac &   &  &     \\
bac &  \frac{1}{3}  & \frac{1}{3} & \frac{1}{12}
\end{array}
\quad
\begin{array}{cccc}
V''' & a & b & c \\ \hline
cba &  0 & 0 & \frac{3}{4} \\
bac &   & &  \\
bac &   &  &     \\
bac &  \frac{1}{3}  & \frac{1}{3} & \frac{1}{12}
\end{array}
\]
Entries for $1$ at $V'''$ if follow from the fact that everyone gets total shares of $\frac{3}{4}$ (otherwise ETE implies that $1$ would receive either strictly more or strictly less and via SP we arrive at a contradiction once we consider a change in $1$'s preferences, aligning with the preferences of others). By XEFF, $1$ gets none of $a$ and $b$. Entries for $1$ at $V''$ follow by SP. Entries for others, in particular $4$ follow as residuals by XEFF and ETE. For $V'$, entries for $4$ follow by XEFF and SP (from $V''$).
\[
\begin{array}{cccc}
W'' & a & b & c \\ \hline
cb &   &  & \\
bca &   &  &  \\
bac &   &  &     \\
ab &    &  &
\end{array}
\qquad
\begin{array}{cccc}
V & a & b & c \\ \hline
cb &   &  & \\
bac    & \frac{1}{6}& & \\
bac    & \frac{1}{6}  & &    \\
ab    & \frac{2}{3} & 0 &
\end{array}
\quad
\begin{array}{cccc}
\hat{V}' & a & b & c \\ \hline
bc &   &  \frac{1}{3} & \\
bac &  \frac{5}{24} & \frac{1}{3}&  \\
bac &  \frac{5}{24} & \frac{1}{3} &     \\
ab   & \frac{7}{12} & 0 & 0
\end{array}
\quad
\begin{array}{cccc}
\hat{V}'' & a & b & c \\ \hline
bc &   &  & \\
bac &   & &  \\
bac &   &  &     \\
abc   & \frac{7}{12} & 0 & \frac{1}{6}
\end{array}
\quad
\begin{array}{cccc}
\hat{V}''' & a & b & c \\ \hline
bca &  0 & \frac{1}{4} & \frac{1}{2} \\
bac &   & &  \\
bac &   &  &     \\
bac &  \frac{1}{3}  & \frac{1}{4} & \frac{1}{6}
\end{array}
\]
Probabilities at $\hat{V}'''$ again follow from the fact that everyone receives total shares of $\frac{3}{4}$, 1 receives none of $a$ by XEFF, and by Lemma 4 (a) as well as ETE for the residual. For $\hat{V}''$ the entries for 4 follow by SP and XEFF. Similarly for $4$ at $\hat{V}'$ and the remaining entries at $\hat{V}'$ by Lemma 4 (b) and as residuals via ETE. To determine entries for $c$ at $\hat{V}'$, consider the following:
 \[
\begin{array}{cccc}
V & a & b & c \\ \hline
cb &   &  & \frac{5}{6} \\
bac    & \frac{1}{6}& & \frac{1}{12}\\
bac    & \frac{1}{6} & & \frac{1}{12}    \\
ab    & \frac{2}{3} & 0 &
\end{array}
\quad
\begin{array}{cccc}
\hat{V}' & a & b & c \\ \hline
bc & 0  & \frac{1}{3} & \frac{1}{2}\\
bac &  \frac{5}{24} & \frac{1}{3}&  \\
bac &  \frac{5}{24} & \frac{1}{3} &     \\
ab   & \frac{7}{12} & 0 & 0
\end{array}
\quad
\begin{array}{cccc}
\tilde{V}'' & a & b & c \\ \hline
bca & 0  & \frac{1}{3} & \frac{1}{2}\\
bac &   & &  \\
bac &   &  &     \\
ab   &  &  &
\end{array}
\quad
\begin{array}{cccc}
\tilde{V}''' & a & b & c \\ \hline
bac &  \frac{1}{6} & \frac{1}{3} & \frac{1}{3} \\
bac &   & &  \\
bac &   &  &     \\
ab &  \frac{1}{2}  & 0 & 0
\end{array}
\quad
\begin{array}{cccc}
\tilde{V}'''' & a & b & c \\ \hline
bac &  & & \\
bac &   & &  \\
bac &   &  &     \\
bac&  \frac{1}{4}  & \frac{1}{4} & \frac{1}{4}
\end{array}
\]

For $\tilde{V}''''$ XEFF and ETE suffice. For $\tilde{V}'''$ XEFF and SP yields entries for $4$, the remainder follows by ETE and XEFF. For 1 at $\tilde{V}''$, the entries follow by XEFF and SP. Similarly, we get a new entry for $1$ at $\hat{V}'$: she receives $c$ with probability $\frac{1}{2}$. But then $1$ receives $c$ with probability $\frac{5}{6}$ at $V$ so that $2$ and $3$ each receive $\frac{1}{12}$ of $c$ at $V$.

Now for $W''$, this allows us to fill in the missing entries by SP: since the total probability mass of $a$ and $c$ remains unchanged for $2$, and since we already know that they receive $\frac{1}{6}$ of $c$ at $W''$, this leaves $\frac{1}{12}$ for object $a$ :
\[
\begin{array}{cccc}
W' & a & b & c \\ \hline
cb &  0 & 0 & \frac{3}{4}-\frac{\alpha}{2} \\
bca &  & \frac{1}{3} & \\
bac &   & \frac{1}{3}  &   \\
ba &  \frac{5}{12}  & \frac{1}{3} &  0
\end{array}
\qquad
\begin{array}{cccc}
W'' & a & b & c \\ \hline
cb &   & 0 & \frac{3}{4}\\
bca &  \frac{1}{12} & \frac{1}{2} & \frac{1}{6} \\
bac &  \frac{1}{6} & \frac{1}{2} & \frac{1}{12} \\
ab &   \frac{3}{4}  & 0 &
\end{array}
\qquad
\begin{array}{cccc}
V & a & b & c \\ \hline
cb &   &  0 & \frac{5}{6}\\
bac    & \frac{1}{6}& \frac{1}{2} & \frac{1}{12} \\
bac    & \frac{1}{6}  & \frac{1}{2}&   \frac{1}{12} \\
ab    & \frac{2}{3} & 0 &
\end{array}
\]
Again, the symmetry of $W''$ allows us to derive the entries for $3$ analogously, so she also gets $\frac{1}{12}$ of their third-most preferred object at $W''$. The most-preferred objects of $1$ and $4$ are assigned with residual probabilities. Finally, returning to $W'$, we know that $4$ is still unassigned with probability $\frac{1}{4}$ by SP.\medskip

Summarizing the three steps, we know that at $W'$, $2$ is unassigned with probability $\frac{1}{4}$, $3$ with probability $\frac{5}{24}$ and $4$ with probability $\frac{1}{4}$. Since unassignment probabilities need to sum to 1 (all 3 objects assigned among 4 agents with probability 1), we know that agent $1$ is unassigned with probability $\frac{7}{24}$. Given that 1 is assigned neither $a$ nor $b$, and assigned $c$ with probability $\frac{3}{4}-\frac{\alpha}{2}$, this implies $\frac{7}{24}=\frac{1}{4}+\frac{\alpha}{2}$ or $\alpha=\frac{1}{12}$, the desired contradiction to $\alpha>\frac{2}{24}$ in the case that $f$ satisfies XEFF.\hfill$\square$

\section{Proof of Lemma \ref{lemm1}}
The proof uses the fact that any (deterministic) serial dicatorship mechanism is non-bossy (NB), i.e., that $f_i^\succ(R)=f_i^\succ(R_i',R_{-i})$ implies $f_j^\succ(R)=f_j^\succ(R_i',R_{-i})$ for all $i,j$, $R$, and $R_i'$.

\noindent
\underline{\textit{(i):}} Let $O=\{a,b,c\}$
and $f$ denote RSD. We consider first $|N|=4$ and second $|N|>4$.\medskip

\noindent
\underline{\textit{$|N|=4$:}} First, consider $|N|=4$ and $N=\{1,2,3,4\}$.
Suppose that $f$ is $q$-agent-object-wasteful where $q> \frac{1}{6}=\frac{4}{24}$ (and $|N|!=24$).

Let $R$ be such that (without loss of generality) $c\in A(R_1)$ and $\min\{f_{11}(R),f_{cc}(R)\}>\frac{4}{24}$.
Then by XEFF, $c\notin top(R)$. Let $top(R_1)=a$. Again by XEFF and 1 being unassigned with positive probability, at least one agent other than 1 ranks $c$ acceptable, say, without loss of generality, $c\in A(R_2)$. Thus, $top(R_2)\neq c$.
If $cP_33$ or $cP_44$, then by XEFF and $|O|=3$ object $c$ is always assigned for any order (as otherwise 1 and 2 must be assigned $a$ and $b$ and $c$ is left over for agent 3 or 4), a contradiction. Hence, both $3P_3c$ and $4P_4 c$. If agent 3 or agent 4 ranks no object acceptable, then we have a contradiction to non-wastefulness of RSD for three agents and three objects (Proposition 2 of Martini (2016)). Thus, $top(R_3),top(R_4)\in \{a,b\}$.
We consider the following two cases where $top(R_1)=top(R_2)$ or $top(R_1)\neq top(R_2)$.\medskip

\noindent
\underline{\textit{Case 1 for $|N|=4$: $top(R_1)\neq top(R_2)$.}} If $top(R_1)\neq top(R_2)$, then $top(R_2)=b$.
For $1$ to be unassigned, $c$ must be assigned to $2$ (so $b$ must be assigned to $3$ or $4$) and $a$ to $3$ or $4$ -- hence $3$ or $4$ must choose $b$ before $2$, $2$ must choose $c$ before $1$ while the other agent among $3$ and $4$ must choose $a$ before $1$. Thus $1$ can be unassigned only for the orders $3421$, $4321$, $4231$ and $3241$, i.e., four out of 24 orders, a contradiction.\medskip

\noindent
\underline{\textit{Case 2 for $|N|=4$: $top(R_1)= top(R_2)=a$.}}
If $R_1:ac\ldots$ and $R_2:ac\ldots$, then by XEFF object $c$ cannot remain unassigned, a contradiction. Hence, either $1$ or $2$ must rank $b$ above $c$. We consider two subcases.

Suppose $2$ ranks $b$ above $c$. For $1$ to be unassigned, $c$ must be assigned to $2$ (so $b$ must be assigned to $3$ or $4$) and $a$ to $3$ or $4$ -- hence $3$ or $4$ must choose $b$ before $2$, $2$ must choose $c$ before $1$ while the other agent among $3$ and $4$ must choose $a$ before $1$. Thus $1$ can be unassigned only for the orders $3421$, $4321$, $4231$ and $3241$, i.e., four out of 24 orders, a contradiction.\medskip

Suppose $1$ ranks $b$ above $c$. For $1$ to be unassigned, $c$ must be assigned to $2$ (so $a$ must be assigned to $3$ or $4$) and $b$ to $3$ or $4$ -- hence $3$ or $4$ must choose $a$ before $2$, $2$ must choose $c$ before $1$ while the other agent among $3$ and $4$ must choose $b$ before $1$. Thus $1$ can be unassigned only for the orders $3421$, $4321$, $4231$ and $3241$, i.e., four out of 24 orders, a contradiction.\medskip

\noindent
\underline{\textit{$|N|>4$:}}
Second, consider $|N|=k+1$ where $k\geq 4$. Let $N=\{1,\ldots,n\}$. Now the above serves as the Induction Hypothesis for four agents. %Note that the proof of Theorem \ref{the1} can always be embedded for such problems where all agents except for $\{1,2,3,4\}$ rank all objects as unacceptable and $f$ is $\frac{1}{6}$-agent-object-wasteful for $|O|=3$ and $|N|\geq 4$. 
In showing the Induction Step, suppose that $f$ is at-most-$\frac{1}{6}$-agent-object-wasteful whenever there are at most $k$ agents.
Suppose that $f$ is $q$-agent-object-wasteful with $q>\frac{1}{6}$ when $|N|=k+1$. Let $R\in \mathcal{R}^N$ be such that
(without loss of generality) $c\in A(R_1)$ and  $q=\min\{f_{11}(R),f_{cc}(R)\}>\frac{1}{6}$ and $c\in A(R_2)$. By XEFF, $top(R_1)\neq c \neq top(R_2)$, and both $iP_i c$ and
$top(R_i)\in \{a,b\}$ for all $i\in N\backslash \{1,2\}$ (where for the latter we also invoke the induction hypothesis). Let $S_a=\{i\in N\backslash \{1,2\}:top(R_i)=a\}$ and $S_b=\{i\in N\backslash \{1,2\}:top(R_i)=b\}$. We consider again the two cases $top(R_1)\neq top(R_2)$ and $top(R_1)=top(R_2)$.\medskip

\noindent
\underline{\textit{Case 1 for $|N|>4$: $top(R_1)\neq top(R_2)$.}}
For the first case, $top(R_1)\neq top(R_2)=b$ and $n\geq 5$. Consider the probability that $c$ remains unassigned. By XEFF we then have that $1$ receives $a$ and $2$ receives $b$. Thus, 1 must choose object $a$ before the agents in $S_a$ and 2 must choose object $b$ before the agents in $S_b$. Now for an arbitrary order 1 chooses $a$ before $S_a$ with probability $\frac{1}{|S_a|+1}$ and 2 chooses $b$ before $S_b$ with probability
$\frac{1}{|S_b|+1}$. Hence, the probability that $c$ is unassigned is bounded from above by
\[ \frac{1}{|S_a|+1}\cdot \frac{1}{|S_b|+1}.\]
As $|S_b|+1=n-|S_a|-1$ we obtain
\begin{equation}\label{eq1}
\frac{1}{|S_a|+1}\cdot \frac{1}{n-|S_a|-1}\leq \frac{1}{6} \Leftrightarrow 6\leq (|S_a|+1)\cdot (n-|S_a|-1)
\end{equation}
where $|S_a|\in \{0,1,\ldots,n-2\}$. Note that the righthandside of (\ref{eq1}) is identical for $|S_a|=0$ and $|S_a|=n-2$, and it is increasing in $|S_a|$ for $|S_a|=0$ and decreasing in $|S_a|$ for $|S_a|=n-2$. Thus, (\ref{eq1}) is always true when (i) $n\geq 7$,
(ii) $n= 6$ and $|S_a|\neq 0,n-2$, and (iii) $n=5$ and $|S_a|\neq 0,n-2$. In all those cases we have $f_{cc}(R)\leq \frac{1}{6}$ and hence,
$\min\{f_{11}(R),f_{cc}(R)\}\leq \frac{1}{6}$.

It remains to consider $n\in \{5,6\}$ and $|S_a|\in \{0,n-2\}$. Let $n=5$ and either $|S_a|=3$ or $|S_a|=0$.

Let $|S_a|=3$ and hence $|S_b|=0$. We consider two subcases. First, suppose that $3$, $4$ and $5$ consider both $a$ and $b$ acceptable. But then $c$ is unassigned only if agents 1 and 2 get objects $a$ and $b$. Thus, 1 and 2 must occupy the first two positions of the order. As there are six orders of the form $12\ldots$ and six orders of the form $21\ldots$, $c$ is wasted with probability $\frac{12}{120}=\frac{1}{10}<\frac{1}{6}$, a contradiction.

Second, suppose that at most two, say $4$ and $5$, consider $b$ acceptable and consider the probability that $1$ is unassigned, for which  $2$ must receive $c$, $4$ or $5$ must receive $b$, and $3$, $4$ or $5$ must receive $a$.\footnote{If fewer than two among $3$, $4$, and $5$ consider $b$ acceptable the probability of that event is even lower.} Thus, $4$ or $5$ must choose before $2$, $2$ must choose before $1$, and $3$, $4$ or $5$ must choose first. In total, there are fifteen such orders:
\begin{itemize}
\item four orders of $N\backslash \{5\}$ of the form $3421$ with $5$ inserted somewhere after $3$;
\item three orders of $N\backslash \{4\}$ of the form $3521$ with $4$ inserted somewhere after $3$ (while not double counting the order $35421$ from the above);
\item four orders of $N\backslash \{3\}$ of the form $4521$ with $3$ inserted somewhere after $4$; and
\item four orders of $N\backslash \{3\}$ of the form $5421$ with $3$ inserted somewhere after $5$.
\end{itemize}
Hence, with $5!=120$ orderings overall, the probability that $1$ remains unassigned is at most $\frac{15}{120}<\frac{20}{120}=\frac{1}{6}$.

Next suppose $|S_a|=0$ and hence $|S_b|=3$. If $3$, $4$ and $5$ consider both $b$ and $a$ acceptable, then the previous, first subcase establishes that this is not possible since the probability that $c$ is unassigned is $\frac{1}{10}<\frac{1}{6}$. If at most two of them, say $4$ and $5$, consider both $b$ and $a$ acceptable, then an argument analogous to the previous, second subcase establishes that this is not possible: for $1$ to be unassigned, $2$ must receive $c$, $4$ or $5$ must receive $a$, and $3$, $4$ or $5$ must receive $b$. Thus, $4$ or $5$ must choose before $2$, $2$ must choose before $1$, and $3$, $4$ or $5$ must choose first. In total, there are fifteen such orders
%\begin{itemize}
%\item four orders of $N\backslash \{5\}$ of the form $3421$ with $5$ inserted somewhere after $3$;
%\item three orders of $N\backslash \{4\}$ of the form $3521$ with $4$ inserted somewhere after $3$ (while not double counting the order $35421$ from the above);
%\item four orders of $N\backslash \{3\}$ of the form $4521$ with $3$ inserted somewhere after $4$; and
%\item four orders of $N\backslash \{3\}$ of the form $5421$ with $3$ inserted somewhere after $5$.
%\end{itemize}
%Hence, with $5!=120$ orderings overall,
and the probability that $1$ remains unassigned is at most $\frac{15}{120}<\frac{20}{120}=\frac{1}{6}$.

This completes the proof for $n=5$ and $|S_a|\in \{0,3\}$. For $n=6$ and $|S_a|\in \{0,4\}$ analogous arguments show that
$\min\{f_{11}(R),f_{cc}(R)\}\leq \frac{1}{6}$.\medskip

\noindent
\underline{\textit{Case 2 for $|N|>4$: $top(R_1)= top(R_2)$.}}
For the second case, $top(R_1)=top(R_2)=a$. As above for $|N|=4$ we consider four subcases.

For the first subcase, if $R_1:ac\ldots$ and $R_2:ac\ldots$, then by XEFF object $c$ cannot remain unassigned, a contradiction.

For the second subcase, if $R_1:ac\ldots$ and $R_2:abc$, then agent 1 is unassigned only if agent 2 gets $c$. Now for any order
$\succ$, we have by SP and NB of $f^{\succ}$:
\begin{equation}\label{eq2}
f_1^{\succ}(R)=1 \quad \& \quad  f_2^{\succ}(R)=c \Leftrightarrow f_1^{\succ}(R_2^{a\leftrightarrow b},R_{-2})=1 \quad \& \quad  f_2^{\succ}(R_2^{a\leftrightarrow b},R_{-2})=c.
\end{equation}
Furthermore, if $c$ is wasted, then 1 obtains $a$ and 2 obtains $b$. Thus, for any such order $\succ$, 1 chooses before 2 and by SP and NB of $f^{\succ}$:
\begin{equation}\label{eq3}
f_1^{\succ}(R)=a \quad \& \quad  f_2^{\succ}(R)=b \Rightarrow f_1^{\succ}(R_2^{a\leftrightarrow b},R_{-2})=a \quad \& \quad  f_2^{\succ}(R_2^{a\leftrightarrow b},R_{-2})=b.
\end{equation}
But now from (\ref{eq2}) we obtain $f_{11}(R)=f_{11}(R_2^{a\leftrightarrow b},R_{-2})$, and
from (\ref{eq3}) we obtain $f_{cc}(R)\leq f_{cc}(R_2^{a\leftrightarrow b},R_{-2})$.
Hence, under profile $(R_2^{a\leftrightarrow b},R_{-2})$ we must have that $f$ is $q$-agent-object-wasteful with $q>\frac{1}{6}$ which is a contradiction to the first case as for this profile we have $top(R_1)=a\neq b= top (R_2^{a\leftrightarrow b})$.

For the third subcase, if $R_1:abc\ldots$ and $R_2:ac\ldots$, we proceed as in the second subcase. First, note that agent 1 is unassigned only if agent 2 gets $c$. Now for any order
$\succ$, we have by SP and NB of $f^{\succ}$:
\begin{equation}\label{eq2'}
f_1^{\succ}(R)=1 \quad \& \quad  f_2^{\succ}(R)=c \Leftrightarrow f_1^{\succ}(R_1^{a\leftrightarrow b},R_{-1})=1 \quad \& \quad  f_2^{\succ}(R_1^{a\leftrightarrow b},R_{-1})=c.
\end{equation}
Furthermore, if $c$ is wasted, then 1 obtains $b$ and 2 obtains $a$. Thus, for any such order $\succ$, 2 chooses before 1 and by SP and NB of $f^{\succ}$:
\begin{equation}\label{eq3'}
f_1^{\succ}(R)=b \quad \& \quad  f_2^{\succ}(R)=a \Rightarrow f_1^{\succ}(R_1^{a\leftrightarrow b},R_{-1})=b \quad \& \quad  f_2^{\succ}(R_1^{a\leftrightarrow b},R_{-1})=a.
\end{equation}
But now from (\ref{eq2'}) we obtain $f_{11}(R)=f_{11}(R_1^{a\leftrightarrow b},R_{-1})$, and
from (\ref{eq3'}) we obtain $f_{cc}(R)\leq f_{cc}(R_1^{a\leftrightarrow b},R_{-1})$.
Hence, under profile $(R_1^{a\leftrightarrow b},R_{-1})$ we must have that $f$ is $q$-agent-object-wasteful with $q>\frac{1}{6}$. But then simply relabelling objects $a$ and $b$ yields a profile in contradiction to the first case.

%For the third subcase, if $R_1:abc$ and $R_2:ac\ldots$, then we show that without loss of generality we may suppose $2P_2b$. If not, then $R_2:acb$ and let $R_2':ac$. Now for any order $\succ$, $c$ remains unassigned under $f^{\succ}(R)$ if and only if $c$ remains unassigned under $f^{\succ}(R_2',R_{-2})$. Furthermore, for any order $\succ$, by SP and NB of $f^{\succ}$ we have
%\begin{equation}
%f_1^{\succ}(R)=1 \quad \& \quad  f_2^{\succ}(R)=c \Leftrightarrow f_1^{\succ}(R_2',R_{-2})=1 \quad \& \quad  f_2^{\succ}(R_2',R_{-2})=c.
%\end{equation}
%Hence, we have $q=\min \{f_{11}(R),f_{cc}(R)\}=\min \{f_{11}(R_2',R_{-2}),f_{cc}(R_2',R_{-2})\}$ where $q>\frac{1}{6}$.
%Thus, without loss of generality, we suppose that $R_2:ac$ (i.e., $2P_2b$).
%But then as $top(R_1)=top(R_2)=a$, we have $f_{11}(R)\leq f_{22}(R)$ and $q=\min\{f_{22}(R),f_{cc}(R)\}$.
%But by exchanging the roles of 1 and 2, we are now in the second subcase above, a contradiction.

For the fourth subcase, if $R_1=R_2:abc$, then again 1 is unassigned only if agent 2 gets $c$, i.e., (\ref{eq2}) continues to hold and
$f_{11}(R)=f_{11}(R_2^{a\leftrightarrow b},R_{-2})$.
Since $\min\{f_{11}(R),f_{cc}(R)\}>\frac{1}{6}$, we must have $f_{11}(R)=f_{11}(R_2^{a\leftrightarrow b},R_{-2})>\frac{1}{6}$. We consider $n=5$, $n=6$ and $n\geq 7$.

For $n=5$, suppose 3, 4 and 5 rank both $a$ and $b$ acceptable. But then 1 remains unassigned when 3, 4 and 5 choose $a$ and $b$ before 2 and 2 chooses $c$ before 1. In total, there are 18 orders of where this is obtained:
\begin{itemize}
\item four orders of $N\backslash \{5\}$ of the form $3421$ with $5$ inserted somewhere after $3$;
\item two orders of $N\backslash \{4\}$ of the form $3521$ with $4$ inserted somewhere after $3$ (while not double counting the orders $34521$ and $35421$ from above);
\item four orders of $N\backslash \{5\}$ of the form $4321$ with $5$ inserted somewhere after $4$;
\item two orders of $N\backslash \{3\}$ of the form $4521$ with $3$ inserted somewhere after $4$ (while not double counting the orders $45321$ and $43521$ from above);
\item four orders of $N\backslash \{4\}$ of the form $5321$ with $4$ inserted somewhere after $5$; and
\item two orders of $N\backslash \{3\}$ of the form $5421$ with $3$ inserted somewhere after $5$ (while not double counting the orders $53421$ and $54321$ from above).
\end{itemize}
Then $f_{11}(R)=\frac{18}{120}<\frac{20}{120}=\frac{1}{6}$.
%But then $c$ remains only unassigned when 1 and 2 get objects $a$ and $b$, and 1 and 2 must choose before 3, 4 and 5. As there are only twelve such orders, we have $f_{cc}(R)=\frac{12}{120}<\frac{1}{6}$, a contradiction. Now if two agents of 3, 4 and 5 rank only one object acceptable, then for at least one of the above orders 1 does not remain unassigned, and 1 is unassigned with probability smaller than or equal to $\frac{1}{6}$, a contradiction. Thus, suppose that 3 ranks only $a$ acceptable while 4 and 5 rank both $a$ and $b$ acceptable. But then $c$ remains unassigned in addition to the above twelve orders (where 1 and 2 choose before 3, 4 and 5) for $13245$, $13254$, $23145$ and $23154$. Hence, $f_{cc}(R)=\frac{16}{20}<\frac{1}{6}$, which is again a contradiction. If $3$ ranks only $b$ acceptable, then by $R_1=R_2:abc$, $c$ remains unassigned only if 1 and 2 choose before 3, 4 and 5, which finishes $n=5$.

Next, let $n\geq 6$. If agents $3$ to $n$ consider only $a$ acceptable while $R_1:abc$, $1$ will never be unassigned.
%If instead $R_1:ac\ldots$, then $2$ must rank $b$ above $c$ (as otherwise by XEFF, $c$ never remains unassigned). But then for $1$ to be unassigned, $2$ must be assigned $c$ and hence agents $3$ to $n$ are assigned $b$, contradicting the assumption that they all consider only $a$ acceptable.

So for the remainder, we can assume that among agents $3$ to $n$ some consider $b$ acceptable (possibly also $a$).
Suppose $1$ and $2$ both rank $c$ below $b$ (i.e., $R_1=R_2:abc$) while $k$ agents with $1\leq k\leq n-2$ among agents $3$ to $n$, say agents $3$ to $2+k$, consider $b$ acceptable and the remaining agents among agents $3$ to $n$, say agents $3+k$ to $n$, consider only $a$ acceptable. Then $c$ is unassigned iff $1$ and $2$ are assigned $a$ and $b$, i.e., if $1$ or $2$ chooses first while the other one of the two chooses before agents $3$ to $2+k$. Overall this has probability at most $\frac{2}{n}\cdot\frac{1}{k+1}\leq\frac{1}{n}\leq \frac{1}{6}$. %If instead (ii') among $(3+k)-6$ some consider also $b$ acceptable this only reduces the probability that $1$ and $2$ are assigned $a$ and $b$, hence reduces the probability that $c$ remains unassigned.
Finally, if $1$ and $2$ rank $c$ higher, this only further reduces the probability that $c$ remains unassigned.

\bigskip

\noindent
\underline{\textit{(ii):}} Let $N=\{1,2,3,4\}$
and $f$ denote RSD. From (i) we know that $f$ is at-most-$\frac{1}{6}$-agent-object-wasteful for $|O|=3$. This serves as the Induction Hypothesis for (ii).

Let $|O|=k+1$ with $k\geq 3$ and $\{a,b,c\}\subseteq O$. Suppose that $f$ is at-most-$\frac{1}{6}$-agent-object-wasteful whenever the set of objects contains at most $k$ elements. Suppose that $f$ is $q$-agent-object-wasteful with $q>\frac{1}{6}$ when $|O|=k+1$. Let $R\in \mathcal{R}^N$ be such that
(without loss of generality) $c\in A(R_1)$, $q=\min\{f_{11}(R),f_{cc}(R)\}>\frac{1}{6}$ and $c\in A(R_2)$. As there are 24 orders and $\frac{1}{6}=\frac{4}{24}$, 1 is unassigned for at least five orders and $c$ remains unassigned for at least five orders. By XEFF, $c\notin top (R)$.
By the induction hypothesis, $\cup_{i\in N}A(R_i)=O$ (as otherwise we may remove any object, which is unacceptable for all agents, from the problem) and $A(R_i)\neq \emptyset$ for all $i\in N$ (as otherwise again by Proposition 2 of Martini (2016) RSD is non-wasteful for three agents).
As agent 1 is unassigned with positive probability, we must have $|A(R_1)|\leq 3$. Let $top(R_1)=a$. By bounded invariance\footnote{This means that for $R_1'$ such that $R_1'|O=R_1|O$ and $A(R_1')=B(c,R_1)$ we have $f_{ic}(R)=f_{ic}(R_1',R_{-1})$ for all $i\in N$.} and SP of RSD, without loss of generality,
we may suppose that $c$ is the last object acceptable under $R_1$.

%Similarly, we may suppose that any $i\in N\backslash \{2,3,4\}$, who ranks $c$ acceptable, does not rank any object in $O\backslash A(R_1)$ %acceptable and below $c$: if $top(R_i)P_icP_idP_ii$ with $d\notin A(R_1)$, then let $R_i'$ be such that $A(R_i')=A(R_i)\backslash \{d\}$ and
%$R_i'|A(R_i')=R_i|A(R_i')$. By BI, $f_{cc}(R_i',R_{-i})=f_{cc}(R)$. Let $\succ$ be such that $f_1^{\succ}(R)=1$: if $1\succ i$ or %$f_i^{\succ}(R)R_ic$, then we continue to have $f_1^{\succ}(R_i',R_{-i})=1$; and if both $i\succ 1$ and $cP_if^{\succ}_{i}(R)$, then the agents %in $N\backslash \{1,i\}$ must have chosen $top(R_i)$ and $c$ before $i$ and we have $f_i^{\succ}(R_i',R_{-i})\in \{f_i^{\succ}(R),d\}$, and we %continue to have $f_1^{\succ}(R_i',R_{-i})=1$. Thus, $f_{11}(R)=f_{11}(R_i',R_{-i}$ and
%$q = \min \{f_{11}(R),f_{cc}(R)\}=\min \{f_{11}(R_i',R_{-i}),f_{cc}(R_i',R_{-i})\}$. Furthermore, similarly, for any $i\in N\backslash %\{2,3,4\}$, we may suppose that $|A(R_i)|\leq 3$: otherwise $i$ only can pick his fourth acceptable object when $i$ is last in the order and 1 is %assigned, and by truncating $i$'s preference by only keeping the first three objects of $R_i$ acceptable, 1 remains unassigned with the same %probability and $c$ remains unassigned with weakly higher probability.

If $|top(R)|=3$, say $top(R)=\{a,b,d\}$, then 1 is unassigned only if 2, 3 and 4 choose $\{a,b,d,c\}$, which is impossible. Thus, $|top(R)|\leq 2$.

Furthermore, suppose that both $cP_33$ and $cP_44$. Then all agents rank $c$ acceptable and object $c$ remains unassigned only if each agent gets an object preferred to $c$, i.e., there exists an allocation $x=(x_1,x_2,x_3,x_4)$ such that $x_iP_ic$ for all $i\in N$. But then for 1 to be unassigned, 2, 3 or 4 must get $c$. Suppose that (without loss of generality as the other cases for 3 and 4 are treated analogously) for the order $\succ$ we have $f_1^{\succ}(R)=1$ and $f_2^{\succ}(R)=c$. By XEFF, then 3 and 4 get $x_1$ and $x_2$. Without loss of generality, suppose that $f^{\succ}(R)=(1,c,x_1,x_2)$ (as the case $f^{\succ}(R)=(1,c,x_2,x_1)$ is treated analogously by exchanging the roles of 3 and 4). Then we must have
$x_1P_3x_3P_3c$, $x_2P_4x_4P_4c$ and $cP_2\{x_3,x_4\}$.
But then whenever for another order $\succ'$ 1 is unassigned and 3 gets $c$, then 2 and 4 must choose $x_1$ and $x_3$, and by $cP_2x_3$, we must have
$f^{\succ'}(R)=(1,x_1,c,x_3)$ and $x_3P_4x_2$, which is impossible as $f_4^{\succ}(R)=x_2$ and $x_3$ remains unassigned under $f^{\succ}(R)$.
Hence, 1 is assigned when agent 3 gets $c$.
Similarly, whenever for another order $\succ''$ 1 is unassigned and 4 gets $c$, then 2 and 3 must choose $x_1$ and $x_4$, and by $cP_2x_4$, we must have
$f^{\succ''}(R)=(1,x_1,x_4,c)$, which is impossible as $f_4^{\succ}(R)=x_2P_4c$ and $x_2$ remains unassigned under $f^{\succ''}(R)$.
Thus, 1 remains unassigned only when 2 gets $c$. Now this only happens when 3 and 4 get $x_1$ and $x_2$, i.e., for the allocations
$(1,c,x_1,x_2)$ and $(1,c,x_2,x_1)$.

If both of them are efficient, then 3 and 4 must agree on the ranking of $x_1$ and $x_2$.
If both $x_1P_3x_2$ and $x_1P_4x_2$, then $(1,c,x_1,x_2)$ can only be obtained for the order
$3421$ whereas $(1,c,x_2,x_1)$ can only be obtained for the order $4321$, i.e., 1 is unassigned for at most two orders, a contradiction.
If both $x_2P_3x_1$ and $x_2P_4x_1$, then $(1,c,x_1,x_2)$ can only be obtained for the orders
$4321$ and $4231$ whereas $(1,c,x_2,x_1)$ can only be obtained for the orders $3421$ and $3241$, i.e., 1 is unassigned for at most four orders, a contradiction.

If only one of them is efficient, then 3 and 4 must disagree on the ranking of $x_1$ and $x_2$.
If both $x_1P_3x_2$ and $x_2P_4x_1$, then $(1,c,x_1,x_2)$ is only efficient and can only be obtained for the orders
$3421$, $4321$ and $4231$, i.e., 1 is unassigned for at most three orders, a contradiction.
If both $x_2P_3x_1$ and $x_1P_4x_2$, then $(1,c,x_2,x_1)$ is only efficient and can only be obtained for the orders
$3421$, $4321$ and $3241$, i.e., 1 is unassigned for at most three orders, a contradiction.

Hence, we have $3P_3c$ or $4P_4c$. We again consider two cases.\medskip

\noindent\underline{\textit{Case 1: $top(R_1)\neq top(R_2)$.}}
If $top(R_1)\neq top(R_2)$, then $|top(R)|=2$. Let $top(R_2)=b$.

If both $3P_3c$ and $4P_4c$, then
for $1$ to be unassigned, $c$ must be assigned to $2$ (so $b$ must be assigned to $3$ or $4$) and $a$ to $3$ or $4$ -- hence $3$ or $4$ must choose $b$ before $2$, $2$ must choose $c$ before $1$ while the other agent among $3$ and $4$ must choose $a$ before $1$. Thus $1$ can be unassigned only for the orders $3421$, $4321$, $4231$ and $3241$, i.e., four out of 24 orders, a contradiction.

Thus, [$cP_33$ and $4P_4c$] or [$3P_3c$ and $cP_44$], say without loss of generality, $cP_33$ and $4P_4c$.
As $|top(R)|\leq 2$, we have $top(R_3),top(R_4)\in \{a,b\}$. Thus, if 1 is unassigned, then 2 or 3 gets $c$.
We consider four different subcases for 3 and 4's rankings over $a$ and $b$.

For the first subcase, let $aP_3b$ and $aP_4b$ which implies $top(R_3)=top(R_4)=a$.
Suppose that for some order $\succ$, 1 is unassigned and 2 gets $c$.
Then 3 and 4 must get $a$ and $b$. Hence we have the assignment $(1,c,a,b)$, obtained for the order $3421$ (only if $b$ ia acceptable for $4$), or $(1,c,b,a)$, obtained for the order $4321$ (only if $b$ is acceptable for $3$).

Suppose that for some order $\succ'$, 1 is unassigned and 3 gets $c$.
Then 2 and 4 must get $a$ and $b$; by XEFF this leaves $(1,b,c,a)$. That assignment is obtained for the orders $2431$ and $4231$. Moreover, it is obtained for the order $4321$ mentioned before (but only if $b$ is not acceptable for $3$). Last, it may be obtained for the order $4312$ -- but only if $R_1,R_2:ac...$ in which case $c$ is never unassigned. Hence, as long as $c$ may be unassigned, in total there at most four orders where 1 is unassigned, a contradiction.

For the second subcase, let $bP_3a$ and $bP_4a$ which implies $top(R_3)=top(R_4)=b$.
Suppose that for the order $\succ$ 1 is unassigned and 2 gets $c$.
Then 3 and 4 must get $a$ and $b$ and $(1,c,a,b)$ is obtained for the orders $4321$ (if $aP_3c$) and $4231$ (if $cP_2a$)
whereas $(1,c,b,a)$ is obtained for the order $3421$ (if $R_4:ba\ldots$) and $3241$ (if $R_2:bc\ldots$). Suppose that for the order $\succ'$ 1 is unassigned and 3 gets $c$.
Then 2 and 4 must get $a$ and $b$ and $(1,a,c,b)$ is obtained for the order $4321$ (if $cP_3a$) and $4231$ (if $aP_2c$)
whereas $(1,b,c,a)$ is obtained for the order $2431$ and $2341$ (if $cP_3a$). As either $aP_2c$ or $cP_2a$, and either $aP_3c$ or $cP_3a$, in total there are at most five orders (if $aP_3c$), respectively, at most six orders (if $cP_3a$) where 1 is unassigned.
If $aP_3c$, then we must have $R_4:ba\ldots$ and $R_2:bc\ldots$, and also 1 does not rank any object $d\in O\backslash \{a,b,c\}$ above $c$.
Now for $c$ to remain unassigned, 1, 2 and 3 must get an object preferred to $c$. Thus, by $R_2:bc\ldots$, 2 gets $b$ and by $A(R_1)\subseteq \{a,b,c\}$, 1 gets $a$ and 3 gets $d\in O\backslash \{a,b,c\}$. But then $dP_3c$ and 3 does not choose $c$ for the order $2431$ and 1 is unassigned for at most four orders, a contradiction.
If $cP_3a$, then $R_3:bc\ldots$ and 1 does not rank any object $d\in O\backslash \{a,b,c\}$ above $c$.
Now for $c$ to remain unassigned, 1, 2 and 3 must get an object preferred to $c$. Thus, by $R_3:bc\ldots$, 3 gets $b$ and by $A(R_1)\subseteq \{a,b,c\}$, 1 gets $a$ and 2 gets $d\in O\backslash \{a,b,c\}$. But then $dP_2c$ and 2 does not choose $c$ for the order $4321$, $4231$, $3421$ and $3241$, and 1 is unassigned for at most four orders, a contradiction.

For the third subcase, let $aP_3b$ and $bP_4a$ which implies $top(R_3)=a$ and $top(R_4)=b$.
Suppose that for the order $\succ$ 1 is unassigned and 2 gets $c$.
Then 3 and 4 must get $a$ and $b$ and $(1,c,a,b)$ is obtained for the orders $3421$, $4321$ and $4231$ (if $cP_2a$)
whereas $(1,c,b,a)$ is inefficient.
Suppose that for the order $\succ'$ 1 is unassigned and 3 gets $c$.
Then 2 and 4 must get $a$ and $b$ and $(1,a,c,b)$ is obtained for the order $4231$ (if $aP_2c$)
whereas $(1,b,c,a)$ is obtained for the order $2431$. As either $aP_2c$ or $cP_2a$, in total there are at most four orders,
%where 1 is unassigned and we must have $R_3:acb\ldots$ and 1 does not rank any object $d\in O\backslash \{a,b,c\}$ above $c$.
%Now for $c$ to remain unassigned, 1, 2 and 3 must get an object preferred to $c$. Thus, by $R_3:acb\ldots$, 3 gets $a$ and by $A(R_1)\subseteq \{a,b,c\}$, 1 gets $b$ and 2 gets $d\in O\backslash \{a,b,c\}$. But then $dP_2c$ and 2 does not choose $c$ for the orders $4321$ and $3421$ and 1 is unassigned for at most three orders,
for which 1 remains unassigned, a contradiction.

For the fourth subcase, let $bP_3a$ and $aP_4b$ which implies $top(R_3)=b$ and $top(R_4)=a$.
Suppose that for the order $\succ$ 1 is unassigned and 2 gets $c$.
Then 3 and 4 must get $a$ and $b$ and $(1,c,b,a)$ is obtained for the orders $3421$, $4321$ and $3241$ (if $cP_2a$)
whereas $(1,c,a,b)$ is inefficient.
Suppose that for the order $\succ'$ 1 is unassigned and 3 gets $c$.
Then 2 and 4 must get $a$ and $b$ and $(1,a,c,b)$ is obtained for the order $4231$ (if $aP_2c$)
whereas $(1,b,c,a)$ is obtained for the orders $2431$ and $2341$ (if $R_3:bc\ldots$). As either $aP_2c$ or $cP_2a$, in total there are at most five orders where 1 is unassigned and we must have $R_3:bc\ldots$ and 1 does not rank any object $d\in O\backslash \{a,b,c\}$ above $c$.
Now for $c$ to remain unassigned, 1, 2 and 3 must get an object preferred to $c$. Thus, by $R_3:bc\ldots$, 3 gets $b$ and by $A(R_1)\subseteq \{a,b,c\}$, 1 gets $a$ and 2 gets $d\in O\backslash \{a,b,c\}$. But then $dP_2c$ and 2 does not choose $c$ for the orders $4321$ and $3421$ and 1 is unassigned for at most three orders, a contradiction. This finishes Case 1.\medskip

\noindent\underline{\textit{Case 2: $top(R_1)= top(R_2)$.}}
For the second case, $top(R_1)=top(R_2)=a$. As $cP_33$, we also must have $top(R_3)=a$ as otherwise we are in Case 1 by exchanging the roles of agents 2 and 3. Using the fact that $|A(R_1)|\leq 3$, we consider five subcases.

For the first subcase, if $R_1:ac\ldots$ and $R_2:ac\ldots$, then by XEFF object $c$ cannot remain unassigned, a contradiction.

For the second subcase, if $R_1:ac\ldots$ and $R_2:abc$, then 1 is unassigned only if 2 or 3 gets $c$. If 2 gets $c$, then 3 and 4 must choose $a$ and $b$ before 2 which can only happen for the orders $4321$ and $3421$. If 3 gets $c$, then $(1,a,c,b)$ is obtained for the orders $2431$, $2341$ (if $cP_3a$) and $4231$ (if $bP_4a$) whereas $(1,b,c,a)$ is obtained for the order $4231$ (if $aP_4b$).
As either $aP_4b$ or $bP_4a$, 1 is unassigned for at most five orders and we must have $cP_3a$ which is impossible as $top(R_3)=a$. Hence, 1 is unassigned for four orders, a contradiction.

For the third subcase, if $R_1:abc$ and $R_2:ac\ldots$, consider the change in the preferences of agent 1, $R_1':bac$. Now, the probability that $1$ remains unassigned is the same under $R$ and under $R'=(R_1',R_{-1})$, while the probability that $c$ remains unassigned is weakly higher under $R'$. Hence, if $\min\{f_{11}(R),f_{cc}(R)\}>\frac{1}{6}$ then $\min\{f_{11}(R'),f_{cc}(R')\}>\frac{1}{6}$. But then consider $R''$ where we relabel $a$ and $b$ compared to $R'$. Clearly $\min\{f_{11}(R''),f_{cc}(R'')\}=\min\{f_{11}(R'),f_{cc}(R')\}>\frac{1}{6}$. But then $R''$ contradicts Case 1.

%For the third subcase, if $R_1:abc$ and $R_2:ac\ldots$, then we have $f_{22}(R)\geq f_{11}(R)$ and by exchanging the roles of 1 and 2 we are in the second subcase, a contradiction.

For the fourth subcase, if $R_1=R_2:abc$, then $c$ is unassigned only if 1 and 2 get $a$ and $b$ and 3 gets $d\in O\backslash \{a,b,c\}$. Hence, by $top(R_3)=a$, we have $aP_3dP_3c$. But then whenever 3 gets $c$ 1 is assigned as otherwise $\{a,b,d\}$ must chosen before 3, which is impossible. Thus, 1 is unassigned only if 2 gets $c$ and 3 and 4 choose $a$ and $b$ before, which can only happen for the orders $3421$ and $4321$, a contradiction.

For the fifth subcase, if $R_1:abc$ and $dP_2c$ for some $d\in O\backslash \{a,b,c\}$, then as above 1 is assigned when 2 gets $c$ and 1 can only remain unassigned when 3 gets $c$ and 2 and 4 get $a$ and $b$. As $top(R_3)=a$ and unless we are in one of the previous subcases above (by exchanging the roles of agents 2 and 3), we cannot have $R_3:ac\ldots$ and $R_3:abc\ldots$. Hence, for some $e\in O\backslash \{a,b,c\}$ we have $eP_3c$. But then again 1 cannot remain unassigned when 3 gets $c$, which is the final contradiction.\hfill$\square$

\section{Proof of Theorem \ref{theo2}}

Let $|N|\geq 4$ and $|O|\geq 3$.%Similarly, to Theorem 1 we provide a common proof and write in italics the arguments specific for (ii).

As XEFF implies XWNW, it suffices to show (i) and (ii)
for $N=\{1,2,3,4\}$ and $O=\{a,b,c\}$ as all agents other than 1, 2, 3 and 4 may rank no object acceptable and all objects other than $a$, $b$ and $c$ are ranked unacceptable in the same order and below $a$, $b$ and $c$.

We use a common proof to show (i) and (ii) where for (i) $f$ satisfies ETAE, SP and XWNW and for (ii) $f$ satisfies ETE, SP and XEFF (where we specify the arguments for (i) and (ii) whenever necessary). Note that $f$ satisfies ETE, SP and XWNW.

Towards a contradiction, suppose that $f$ is at-most-$q$-object-wasteful with $q<\frac{1}{4}$.
Let
\[\begin{array}{cccc}
C'' & a & b & c \\ \hline
abc & \frac{1}{2} & \alpha_1 &  0\\
cba & 0 & \alpha_2 &  \frac{1}{2}\\
a & \frac{1}{2} & 0 &  0\\
c & 0 & 0 &  \frac{1}{2}
\end{array}.
\]
In determining $f(C'')$, we have $f_{3a}(C'')\leq \frac{1}{2}$ as otherwise agent 3 may report $C_1''$ and we obtain a contradiction to ETE, SP and feasibility. Similarly, we obtain $f_{1a}(C'')\leq \frac{1}{2}$. Then $f_{2a}(C'')=0$ where for (i) this follows from Lemma \ref{lem.none.at.bottom} and for (ii) from XEFF. Hence, by XWNW, we obtain $f_{1a}(C'')=\frac{1}{2}=f_{3a}(C'')$.
Analogously we obtain $f_{2c}(C'')=\frac{1}{2}=f_{4c}(C'')$.
If both $\alpha_1\leq \frac{3}{8}$ and $\alpha_2\leq \frac{3}{8}$, then $f$ is $q$-object-wasteful with $q\geq \frac{1}{4}$, a contradiction.
Thus, $\alpha_1>\frac{3}{8}$ or $\alpha_2>\frac{3}{8}$. Without loss of generality, let $\frac{3}{8}<\alpha_2\equiv \alpha$.
Let
\[\begin{array}{cccc}
C''' & a & b & c \\ \hline
abc & \frac{1}{3} & \frac{1}{2}& \alpha-\frac{1}{3} \\
abc & \frac{1}{3} &  \frac{1}{2} & \alpha-\frac{1}{3}\\
a & \frac{1}{3} & 0 &  0\\
c & 0 & 0 &  \frac{5}{3}-2\alpha
\end{array}.
\]
In determining $f(C''')$, we obtain $f_{1a}(C''')=f_{2a}(C''')=f_{3a}(C''')=\frac{1}{3}$ for (i) by Lemma \ref{lem.equal.split.at.top} (b) and for (ii) by SP and ETE as 3 may report $C_1'''$.
We have $f_{1b}(C''')=\frac{1}{2}=f_{2b}(C''')$ where for (i) by ETAE and XWNW, $b$ must be fully assigned and $f_{1b}(C''')=\frac{1}{2}=f_{2b}(C''')$ (as otherwise agents 1 and 2 may drop $c$) and for (ii) this follows from XEFF and ETE.
As $C''$ and $C'''$ only differ in 2's preference and by ETE, we now obtain
$f_{2c}(C''')= \alpha- \frac{1}{3}= f_{1c}(C''')$. By XWNW, $c$ must be fully assigned and we obtain
$f_{4c}(C''')=\frac{5}{3}-2\alpha$. Note that by $\alpha> \frac{3}{8}$, we have $\alpha-\frac{1}{3}> \frac{1}{24}$ and
$f_{4c}(C''')<\frac{22}{24}$.

Let
\[\begin{array}{cccc}
R^0 & a & b & c \\ \hline
abc & \frac{1}{3} & \frac{1}{2}& \alpha-\frac{1}{3} \\
abc & \frac{1}{3} &  \frac{1}{2} & \alpha-\frac{1}{3}\\
a & \frac{1}{3} & 0 &  0\\
cab & 0 & 0 &  \frac{5}{3}-2\alpha
\end{array}.
\]
We obtain $f_{1a}(R^0)=f_{2a}(R^0)=f_{3a}(R^0)=\frac{1}{3}$ for (i) from Lemma \ref{lem.equal.split.at.top} (b) and for (ii) by XEFF $f_{4a}(R^0)=0$ and from ETE and SP (as 3 may report $R^0_1$).
Furthermore, $f_{4b}(R^0)=0$ for (i) from Lemma \ref{lem.none.at.bottom} and for (ii) from XEFF.
Thus, by ETE and XWNW,
$f_{1b}(R^0)=f_{2b}(R^0)=\frac{1}{2}$.% (as otherwise agents 1 and 2 may drop $c$ and using SP, ETAE and XWNW, agent 3 is assigned with positive probability object $b$, and then agent 3 may report only object $c$ acceptable, which 3 then receives with probability smaller than 1, a contradiction to XWNW as 3 is the only agent for whom $c$ is acceptable).
Thus, $f_{4a}(R^0)=f_{4b}(R^0)=0$. As $R^0$ and $C'''$ only differ in 4's preference, by SP we have $f_4(R^0)=f_4(C''')$.
In particular, $f_{4a}(R^0)+f_{4b}(R^0)+f_{4c}(R^0)=f_{4c}(C''')<\frac{22}{24}$.

Let
\[\begin{array}{cccc}
E & a & b & c \\ \hline
abc & \frac{1}{4} & \frac{1}{3}& \frac{1}{3} \\
abc & \frac{1}{4} &  \frac{1}{3} & \frac{1}{3}\\
a & \frac{1}{4} & 0 &  0\\
abc & \frac{1}{4} & \frac{1}{3} &  \frac{1}{3}
\end{array}.
\]
We obtain $f_{3a}(E)=\frac{1}{4}$ from SP and ETE (as for $E'=(abc,E_{-3})$ object $a$ is shared equally) and the other entries of $f(E)$ using ETE and XWNW.

Note that $E$ and $R^0$ only differ in 4's preference, and both $E_4$ and $R_4^0$ rank all objects acceptable.
But this is a contradiction to SP as
\[ f_{4a}(E)+f_{4b}(E)+f_{4c}(E)=\frac{1}{4}+\frac{1}{3}+\frac{1}{3}=\frac{22}{24}>f_{4c}(R^0)=f_{4a}(R^0)+f_{4b}(R^0)+f_{4c}(R^0).
\]
Hence, we must have $q\geq \frac{1}{4}$ for at-most-$q$-object-wastefulness of $f$.%\medskip
\hfill$\square$

\section{Proof of Lemma \ref{lem.RSD.tight}}

Let $N=\{1,2,3,4,\ldots,n\}$, $O=\{a,b,c\}$, and $f$ denote the RSD-mechanism. As $f$ satisfies \SETE, \SP\ and \XEFF, by Theorem \ref{theo2} it suffices to show that $f$ is not $q$-object-wasteful for $q>\frac{1}{4}$.

Thus, suppose that, without loss of generality, $f_{cc}(R)\geq q>\frac{1}{4}$ and $\sum_{i:cP_ii}f_{ii}(R)\geq q >\frac{1}{4}$ for profile $R$. But then two agents rank $c$ acceptable (as when three or four agents rank $c$ acceptable, then $c$ is fully assigned, and when only one agent ranks $c$ acceptable, then this agent is fully assigned) and all other agents rank $c$ unacceptable, say
$cP_11$, $cP_22$, and  $kP_3c$ for all $k>2$. Moreover $c$ is not ranked first by $1$ or $2$ as it would then be assigned with probability $1$. Denote the set of agents other than $1$ and $2$ who rank $a$ first as $S_a$ and define $S_b$ analogously.

We first consider $n=4$. As $4!=24$ and $f_{cc}(R)> \frac{1}{4}=\frac{6}{24}$, object $c$ is unassigned for at least seven orders.
Now when $c$ is unassigned, then 1 and 2 must be assigned $a$ and $b$, i.e., $c$ is unassigned only for the assignments
\[
\mu = \left(\begin{array}{cccc}
1 & 2 & 3 & 4\\
a & b & 3 & 4
\end{array}\right)
\mbox{ and }
\nu=
\left(\begin{array}{cccc}
1 & 2 & 3 & 4\\
b & a & 3 & 4
\end{array}\right).
\]
If $A(R_3)=\emptyset$ or $A(R_4)=\emptyset$, then RSD is ex-ante non-wasteful by Martini (2016, Proposition 2).
Thus, $top(R_3),top(R_4)\in \{a,b\}$. Without loss of generality, let $top(R_1)=a$. We consider the following three cases.\medskip

\noindent
\underline{\textit{Case 1: $[a\in A(R_3)\, \&\, b\in A(R_4)]$.}}

But then $c$ can be unassigned only if agent 1 chooses first, chooses $top(R_1)=a$, and $2$ chooses $b$ before $4$ -- which arises with probability $\frac{1}{4}\cdot \frac{1}{2}=\frac{1}{8}$ -- or if agent 2 chooses first, chooses $top(R_2)\in \{a,b\}$, and agent 1 chooses the other object of $\{a,b\}$ before agent 3 or agent 4 (whoever among $3$ and $4$ considers it acceptable) -- which arises with probability $\frac{1}{4}\cdot\frac{1}{2}=\frac{1}{8}$. Thus, $c$ is unassigned with probability at most $\frac{1}{8}+\frac{1}{8}=\frac{1}{4}$,  a contradiction.\medskip

\noindent
\underline{\textit{Case 2: $[b\in A(R_3) \, \&\, a\in A(R_4)]$.}}

Analogous to Case 1.\medskip

\noindent
\underline{\textit{Case 3: $A(R_3)=A(R_4)=\{a\}$.}}

If $b\in A(R_1)$, then agent 1 cannot be unassigned and $q$-object-wastefulness reduces to $q$-agent-object-wastefulness, a contradiction as $f$ is at-most-$\frac{1}{6}$-agent-object-wasteful. Thus, $b\notin A(R_1)$ and $R_1:ac$. But then 1 is only unassigned when 3 or 4 receive $a$ and 2 receives $c$, i.e., we must have $cP_2b$. But then $top(R_2)=a$ and both $R_1:ac$ and $R_2:ac\cdots$, and \XEFF\ of $f$ implies that $c$ is not wasted, a contradiction.\medskip

\noindent
\underline{\textit{Case 4: $A(R_3)=A(R_4)=\{b\}$.}}

But then 1 is only unassigned if 2 chooses $a$ before 1, i.e., $a\in A(R_2)$. Hence, 1 cannot be unassigned (as 2 cannot choose both $a$ and $c$ before 1)
and $q$-object-wastefulness reduces to $q$-agent-object-wastefulness, a contradiction as $f$ is at-most-$\frac{1}{6}$-agent-object-wasteful.\smallskip

For $n\geq 5$, as we have dealt with $n\leq 4$ above, w.l.o.g. we may suppose that each agent ranks at least one object acceptable. We consider two cases.\smallskip

\noindent
\underline{\textit{Case I: $top(R_1)\neq top(R_2)$.}}

W.l.o.g., $top(R_1)=a$ and $top(R_2)=b$. Consider the probability that $c$ remains unassigned. By XEFF we then have that $1$ receives $a$ and $2$ receives $b$. Thus, 1 must choose object $a$ before the agents in $S_a$ and 2 must choose object $b$ before the agents in $S_b$. Now for an arbitrary order 1 chooses $a$ before $S_a$ with probability $\frac{1}{|S_a|+1}$ and 2 chooses $b$ before $S_b$ with probability
$\frac{1}{|S_b|+1}$. Hence, the probability that $c$ is unassigned is bounded from above by
\[ \frac{1}{|S_a|+1}\cdot \frac{1}{|S_b|+1}.\]
As $|S_b|+1=n-|S_a|-1$ we obtain
\begin{equation}
\frac{1}{|S_a|+1}\cdot \frac{1}{n-|S_a|-1}\leq \frac{1}{4} \Leftrightarrow 4\leq (|S_a|+1)\cdot (n-|S_a|-1)
\end{equation}
where $|S_a|\in \{0,1,\ldots,n-2\}$. Note that the righthand side is identical for $|S_a|=0$ and $|S_a|=n-2$, and it is increasing in $|S_a|$ for $(|S_a|+1)< (n-|S_a|-1)$ and decreasing in $|S_a|$ for $(|S_a|+1)> (n-|S_a|-1)$. Thus,  it is always true when $n\geq 5$.

\noindent
\underline{\textit{Case II: $top(R_1)= top(R_2)$.}}

W.l.o.g., $top(R_1)=top(R_2)=a$. If agents $3$ to $n$ consider only $a$ acceptable while $R_1:abc$, $1$ will never be unassigned.
If instead $R_1:ac\ldots$, then $2$ must rank $b$ above $c$ (as otherwise by XEFF, $c$ never remains unassigned). But then for $1$ to be unassigned, $2$ must be assigned $c$ and hence agents $3$ to $n$ are assigned $b$, contradicting the assumption that they all consider only $a$ acceptable.

So for the remainder, we can assume that among agents $3$ to $n$ some consider $b$ acceptable (possibly also $a$).
Suppose $1$ and $2$ both rank $c$ below $b$ (i.e., $R_1=R_2:abc$) while $k$ agents with $1\leq k\leq n-2$ among agents $3$ to $n$, say agents $3$ to $2+k$, consider $b$ acceptable and the remaining agents among agents $3$ to $n$, say agents $3+k$ to $n$, consider only $a$ acceptable. Then $c$ is unassigned iff $1$ and $2$ are assigned $a$ and $b$, i.e., if $1$ or $2$ chooses first while the other one of the two chooses before agents $3$ to $2+k$. Overall this has probability at most $\frac{2}{n}\cdot\frac{1}{k+1}\leq\frac{1}{n}\leq \frac{1}{5}$. %If instead (ii') among $(3+k)-6$ some consider also $b$ acceptable this only reduces the probability that $1$ and $2$ are assigned $a$ and $b$, hence reduces the probability that $c$ remains unassigned.
Finally, if $1$ and $2$ rank $c$ higher, this only further reduces the probability that $c$ remains unassigned.\hfill$\square$

\section{Proof of Lemma \ref{lemRDA}}

Let $N=\{1,2,3,4\}$ and $O=\{a,b,c\}$. Let $g$ denote RDA.

For the proof of (a), towards a contradiction, suppose that $g$ is $q$-agent-object-wasteful where $q\geq \frac{1}{6}$.
Let $R$ be such that (without loss of generality), $top(R_1)=a$, $c\in A(R_1)$ and $\min\{g_{11}(R),g_{cc}(R)\}\geq \frac{1}{6}$. 
Then by XNW of $g$, $c\notin top(R)$. Again by XNW and 1 being unassigned with positive probability, at least one agent other than 1 ranks $c$ acceptable, say, without loss of generality, $c\in A(R_2)$. Thus, $top(R_2)\neq c$.
If $cP_33$ or $cP_44$, then by XNW and $|O|=3$ object $c$ is always assigned under RDA, a contradiction. Hence, both $3P_3c$ and $4P_4 c$. If agent 3 or agent 4 ranks no object acceptable, then we have a contradiction to ex-post non-wastefulness of RDA for three agents and three objects (as for $top(R_1)\neq top(R_2)$ or $|A(R_1)|=3$ agent 1 cannot be unassigned; for $top(R_1)=top(R_2)=a$ and $R_1=R_2:ac$ object $c$ cannot be unassigned by XNW, and otherwise we have $R_1:ac$ and $R_2:abc$ and by stability agent 1 cannot be unassigned as then 2 has to receive $c$). Thus, $top(R_3),top(R_4)\in \{a,b\}$.
We consider the following two cases where $top(R_1)=top(R_2)$ or $top(R_1)\neq top(R_2)$.\medskip

\noindent
\underline{\textit{Case 1: $top(R_1)= top(R_2)=a$.}}

\noindent
\textit{Case 1.1: $R_1=R_2:ac\cdots$.}

Then by XNW of $g$, $c$ is always assigned, a contradiction.\smallskip

\medskip

For the remaining subcases of Case 1, we will show that the contradiction arises as $1$ is unassigned with probability smaller than $\frac{1}{6}$. Note that then either 1 or 2 ranks $b$ second after $a$.

\noindent
\textit{Case 1.2: $R_1:abc$ and $aP_2cP_22$.}

Note that $1$ is unassigned only under the following two assignments:
\[
\mu = \left(\begin{array}{cccc}
1 & 2 & 3 & 4\\
1 & c & a & b
\end{array}\right)
\mbox{ and }
\nu=
\left(\begin{array}{cccc}
1 & 2 & 3 & 4\\
1 & c & b & a
\end{array}\right).
\]
But $\mu$ is chosen with probability at most $\frac{1}{12}$ as $2>_c 1$ arises with probability $\frac{1}{2}$,
$3>_a\{1,2\}$ arises with probability $\frac{1}{3}$ and $4>_b 1$ arises with probability $\frac{1}{2}$. Similarly, $\nu$ is chosen with at most probability $\frac{1}{12}$ as $2>_c 1$ arises with probability $\frac{1}{2}$,
$4>_a\{1,2\}$ arises with probability $\frac{1}{3}$ and $3>_b 1$ arises with probability $\frac{1}{2}$. Since we counted double the event $\{3,4\}>_a\{1,2\}$, we have $g_{11}(R)<\frac{1}{6}$, a contradiction.\smallskip

\noindent
\textit{Case 1.3: $aP_1cP_11$ and $R_2:abc$.}

Again, in order for $1$ to be unassigned $2$ needs to be assigned $c$, so $a$ and $b$ need to be assigned to $3$ and $4$ -- which yields $\mu$ or $\nu$. But $\mu$ is chosen with probability at most $\frac{1}{12}$ as $2>_c 1$ arises with probability $\frac{1}{2}$,
$3>_a\{1,2\}$ arises with probability $\frac{1}{3}$ and $4>_b 1$ arises with probability $\frac{1}{2}$. Similarly, $\nu$ is chosen with at most probability $\frac{1}{12}$. Since we counted double the event $\{3,4\}>_a\{1,2\}$, we have $g_{11}(R)<\frac{1}{6}$, a contradiction.\medskip

\noindent
\underline{\textit{Case 2: $top(R_1)=a\neq b= top(R_2)$.}}

We consider three subcases.\smallskip

\noindent
\textit{Case 2.1: $R_1:ac\cdots.$ and $R_2:bc\cdots$.}

First, suppose that both $\mu$ and $\nu$ belong to the support of $g(R)$. Then both $3$ and $4$ need to consider both $a$ and $b$ acceptable. In this case consider the probability that $c$ remains unassigned, which happens only if $1$ is assigned $a$ and $2$ is assigned $b$. But for this to happen, we need $1>_a\{3,4\}$ and $2>_b\{3,4\}$, which arises with probability $\frac{1}{3}\cdot \frac{1}{3}=\frac{1}{9}<\frac{1}{6}$, a contradiction.\smallskip

Second, if $\mu$ belongs to the support of $g(R)$ but $\nu$ does not belong to the support of $g(R)$, then 1 is unassigned with at most probability $\frac{1}{8}$ as $2>_c1$ arises with probability $\frac{1}{2}$, $3>_a1$ arises with probability $\frac{1}{2}$ and $4>_b 2$ arises with probability $\frac{1}{2}$. The case where $\nu$ belongs to the support of $g(R)$ but $\mu$ does not belong to the support of $g(R)$ is analogous.

\noindent
\textit{Case 2.2: $R_1:abc$ and $bP_2cP_22$.}

Again 1 is unassigned only at $\mu$ and $\nu$. But $\mu$ is chosen at most with probability  $\frac{1}{12}$ as $2>_c 1$ arises with probability $\frac{1}{2}$,
$3>_a 1$ arises with probability $\frac{1}{2}$ and $4>_b\{1,2\}$ arises with probability $\frac{1}{3}$, and similarly, $\nu$ has at most probability $\frac{1}{12}$. As we counted double the event $\{3,4\}>_b \{1,2\}$, we have $g_{11}(R)<\frac{1}{6}$, a contradiction.\smallskip

\noindent
\textit{Case 2.3: $aP_1cP_11$ and $R_2:bac$.}

Again 1 is unassigned (only) at $\mu$ and $\nu$.
But $\mu$ is chosen  with probability at most $\frac{1}{12}$ as (by stability) $2>_c 1$ arises with probability $\frac{1}{2}$,
$3>_a\{1,2\}$ arises with probability $\frac{1}{3}$ and $4>_b 1$ arises with probability $\frac{1}{2}$, and similarly, $\nu$ has at most probability $\frac{1}{12}$. As we counted double the event $\{3,4\}>_a \{1,2\}$, we have $g_{11}(R)<\frac{1}{6}$, a contradiction.

For the proof of (b), consider the following profile:
\[
\begin{array}{cccc}
\hat{R} & a & b & c \\ \hline
ac &  \frac{1}{2} & 0 & \frac{3}{8} \\
bc & 0 & \frac{1}{2} & \frac{3}{8}\\
a & \frac{1}{2} & 0 & 0 \\
b & 0 & \frac{1}{2} &  0
\end{array}
\]
For profile $\hat{R}$, object $c$ is unassigned when agent 1 is assigned object $a$ (which arises for $1>_a3$ with probability $\frac{1}{2}$) and agent 2 is assigned object $b$ (which arises for $2>_b4$ with probability $\frac{1}{2}$). Thus, $RDA_{cc}(\hat{R})=\frac{1}{4}$ and $RDA_{1c}(\hat{R})=RDA_{2c}(\hat{R})= \frac{3}{8}$.\footnote{This follows from the fact that RDA satisfies anonymity and neutrality, i.e., the names of the agents and the objects do not matter.}
Hence, $RDA(\hat{R})$ is $\frac{1}{4}$-object-wasteful, and $RDA$ is $q$-object-wasteful with $q\geq \frac{1}{4}$.\hfill$\square$

\section{Erdil reconsidered}\label{Erdil reconsidered}

Let us reconsider Erdil's mechanism and implement the $-\epsilon = - \frac{3}{24}$-worsening of his mechanism (instead of the $\frac{2}{24}$-improvement of his mechanism).

Let $N=\{1,2,3,4\}$ and $O=\{a,b,c\}$ (which can be generalized to arbitrary $N$ and $O$ as in Erdil (2014)).

Let $R_{-1}$ be such that
\[
\begin{array}{c|c|c}
R_2 & R_3 & R_4\\
\hline

c & c & c\\
a& b &
\end{array}
\]
The following three sets partition $\mathcal{R}^1$:
\begin{eqnarray*}
\mathcal{R}^1_a & =\{ R_1\in \mathcal{R}^1: aP_1bP_11\},\\
\mathcal{R}^1_b & =\{ R_1\in \mathcal{R}^1: bP_1aP_11\},\\
\hat{\mathcal{R}}^1 & =\{ R_1\in \mathcal{R}^1: 1P_1 a \mbox{ or } 1P_1b\}.
\end{eqnarray*}
Before formally defining the mechanism,
let $R_1:ab$ and $R=(R_1,R_{-1})$. Note that there are $4!=24$ orders.

Now 1 is assigned $b$ for
\[
\mu = \left(\begin{array}{cccc}
1 & 2 & 3 & 4\\
b & a & c & 4
\end{array}\right)
\mbox{ and }
\nu=
\left(\begin{array}{cccc}
1 & 2 & 3 & 4\\
b & a & 3 & c
\end{array}\right).
\]
Note that $\mu$ is obtained for the orders 3-2-1-4, 3-2-4-1, 3-4-2-1 while $\nu$ is obtained for the order 4-2-1-3, and for 1 to obtain $b$ it suffices that $bP_11$, i.e., the preference between $a$ and $b$ does not play any role. Thus, $RSD_{1b}(R)\geq \frac{4}{24}$.

Next 1 is unassigned for
\[
\eta = \left(\begin{array}{cccc}
1 & 2 & 3 & 4\\
1 & a & b & c
\end{array}\right)
\]
and $\eta$ is obtained for the orders 4-2-3-1 and 4-3-2-1, and 1's preference does not play any role, i.e., $RSD_{11}(R)\geq \frac{2}{24}$.
Now 3 is unassigned for
\[
\eta' = \left(\begin{array}{cccc}
1 & 2 & 3 & 4\\
b & a & 3 & c
\end{array}\right)
\]
and $\eta'$ is obtained for the orders 4-2-1-3 and for this $bP_11$ suffices, i.e., $RSD_{33}(R)\geq \frac{1}{24}$.

Finally, under profile $R$ object $b$ is unassigned for
the orders 3-1-2-4, 1-3-2-4 and 3-4-1-2 (for which we obtain $(a,2,c,4)$) and for this $aP_1bP_1c$ suffices, i.e., $RSD_{bb}(R)\geq \frac{3}{24}$.

Below we then subtract for profile $R$ for agent 1 the share of $\frac{3}{24}$ from object $b$ (and add this to 1's unassigned part) while keeping the assignments of other objects for 1 and the random assignments of other agents unchanged. Then in the modified Erdil-mechanism $g$ for profile $R$, object $b$ is unassigned with probability greater than or equal to $\frac{5}{24}$, agent 1 is unassigned with probability greater than or equal to $\frac{5}{24}$ (and agent 3 is unassigned with probability greater than or equal to $\frac{1}{24}$), which means that $g(R)$ is $\frac{5}{24}$-agent-object-wasteful.
%(I also thought that it is $\frac{7}{24}$-object-wasteful by reducing further by $\frac{1}{24}$ the share of $b$ for agent 1 but I could not finish the argument).

Next we define $g$. Let $Q\in \mathcal{R}^N$.

First, $g_i(Q)=RSD_i(Q)$ for all $Q$ and all $i\neq 1$.

Second, $g(Q)=RSD(Q)$ if $Q_{-1}\neq R_{-1}$ or $Q_1\in \hat{\mathcal{R}}^1$.

Third, let $Q_{-1}=R_{-1}$ and $Q_1\in \mathcal{R}^1_a\cup \mathcal{R}^1_b$.
Note that 1's total assignment equals $\frac{22}{24}$ at $Q$ under RSD, i.e., $RSD_{11}(Q)=\frac{2}{24}$. Then we reduce agent 1's assignment as follows:
\begin{eqnarray*}
g_1(Q) & =RSD_1(Q) - \frac{3}{24}b \mbox{ if } Q_1\in \mathcal{R}^1_a\\
g_1(Q) & =RSD_1(Q) - \frac{3}{24}a \mbox{ if }Q_1\in \mathcal{R}^1_b.
\end{eqnarray*}
Note that in the above then being unassigned is increased by $\frac{3}{24}$ for agent 1.

Next note that for $Q_1\in \mathcal{R}^1_a$, we have shown above for $bQ_11$ that $RSD_{1b}(Q)\geq \frac{4}{24}$, i.e., the above reduction is feasible (and similarly for $Q_1\in \mathcal{R}^1_b$).

Again for $Q_1\in \mathcal{R}^1_a$ the reduction is possible as we do the following
\[
\left(\begin{array}{cccc}
1 & 2 & 3 & 4\\
b & a & c & 4
\end{array}\right)
\mbox{ + }
\left(\begin{array}{cccc}
1 & 2 & 3 & 4\\
a & 2 & b & c
\end{array}\right)
\longrightarrow
\left(\begin{array}{cccc}
1 & 2 & 3 & 4\\
1 & a & b & c
\end{array}\right)
\mbox{ + }
\left(\begin{array}{cccc}
1 & 2 & 3 & 4\\
a & 2 & c & 4
\end{array}\right).
\]
Note that both $\left(\begin{array}{cccc}
1 & 2 & 3 & 4\\
b & a & c & 4
\end{array}\right)$ and $\left(\begin{array}{cccc}
1 & 2 & 3 & 4\\
a & 2 & b & c
\end{array}\right)$ are obtained for three orders (for the first one as shown above, and for the second one for the orders 4-1-3-2, 4-3-1-2 and 4-1-2-3). Note that all these assignments are ex-post efficient under $Q$. Hence, these probability transfers are possible in $RSD(R)$ to obtain $g(R)$ while maintaining XEFF of $g$.

Now SP follows as in Erdil (2014) (and if $Q_1\in \hat{\mathcal{R}}^1$, then the lowest ranked object becomes agent 1's outside option):\medskip

If $aQ_1bQ_11Q_1c$, $aQ_1bQ_1cQ_11$, $aQ_1cQ_1bQ_11$, then $g_1(Q)=(18a+1b)/24$.

If $bQ_1aQ_11Q_1c$, $bQ_1aQ_1cQ_11$, $bQ_1cQ_1aQ_11$, then $g_1(Q)=(18b+1a)/24$.

If $cQ_1aQ_1bQ_11$, then $g_1(Q)=(6c+12a+1b)/24$.

If $cQ_1bQ_1aQ_11$, then $g_1(Q)=(6c+12b+1a)/24$.\medskip

The interesting fact is that $g$ satisfies ETAE (and no randomization is necessary).
% (and we do not care of anonymity or neutrality, remember that $q$-object-wastefulness may not carry over for randomizations of mechanisms).

Now let $Q$ be a profile. If $g(Q)=RSD(Q)$, then ETAE is obvious. Thus, let $g(Q)\neq RSD(Q)$. But then we must have $Q_{-1}=R_{-1}$ and $Q_1\notin \hat{\mathcal{R}}^1$. But then given $R_{-1}$, for ETAE to apply we must have $Q_1:cab$ or $Q_1:cba$.
If $Q_1:cab$, then by definition we have $g_{1c}(Q)=RSD_{1c}(Q)$ and $g_{1a}(Q)=RSD_{1a}(Q)$ and ETAE is satisfied, and similarly for $Q_1:cba$. This yields the desired conclusion.
\end{appendix}

\section*{References}
\begin{description}
\item Abdulkadiro\u{g}lu, Atila, Parag A. Pathak and Alvin E. Roth (2009): Strategy-Proofness versus Efficiency in Matching with Indifferences: Redesigning the
NYC High School Match. American Economic Review 99:1954--1978.
\item Abdulkadiro\u{g}lu, Atila and Tayfun S\"onmez (2003): Ordinal efficiency and dominated sets of assignments.
Journal of Economic Theory 112:157--172.
\item Arnosti, Nick (2023): Lottery design for school choice, Management Science
69:244--259.
\item Ashlagi, Itai and Afshin Nikzad (2020): What matters in school choice tie-breaking? How
competition guides design, Journal of Economic Theory 190:105120.
\item Ashlagi, Itai, Afshin Nikzad, and Assaf Romm (2019): Assigning more students to their top
choices: A comparison of tie-breaking rules, Games and Economic Behavior 115:167--187.
\item Basteck, Christian (2018): Fair solutions to the random assignment problem.
Journal of
Mathematical Economics 79:163--172.
\item Basteck, Christian (2024): An axiomatization of the random priority rule. WZB Discussion Paper SP II 2024-201.
\item Basteck, Christian and Lars Ehlers (2023): Strategy-Proof and Envyfree Random Assignment.
Journal of Economic Theory 209:105618.
\item Basteck, Christian and Lars Ehlers (2025): On the (Constrained) Efficiency of Strategy-Proof Random Assignment.
Econometrica 93: 569--595.
\item Birkhoff, Garrett (1946): Three observations on linear algebra. Univ. Nac. Tacuman, Rev.
Ser. A, 5:147--151.
\item Bogomolnaia, Anna and Herv\'e Moulin (2001): A new solution to the random assignment
problem. Journal of Economic Theory, 100:295--328.
\item Bogomolnaia, Anna and Herv\'e Moulin (2015): Size versus fairness in the assignment problem.
Games and Economic Behavior, 90:119--127.

\item Bade, Sophie (2020): Random serial dictatorship: the one and only. Mathematics of
Operations Research 45(1):353--368.

\item Demeulemeester, Tom and Juan Pereyra (2024): Rawlsian assignments. Working Paper, arXiv preprint arXiv:2207.02930.

\item Duddy, Conal (2025): Egalitarian random assignment.
Economic Theory 80: 321--354.

\item Erdil, Aytek (2014): Strategy-proof stochastic assignment. Journal of Economic Theory,
151:146--162.
\item Han, Xiang (2024a): On the efficiency and fairness of deferred acceptance with single tie-breaking.
Journal of Economic Theory 218:105842.
\item Han, Xiang (2024b): A theory of fair random allocation under priorities.
Theoretical Economics 19:1185--1221.

\item Heo, Eun Jeong, Vikram Manjunath and Samson Alva (2025): Unambiguous efficiency of random
allocations, Working Paper.
\item Kesten, Onur and Utku M. Ünver (2015): A theory of school choice lotteries.
Theoretical
Economics 10:543--595.
\item Martini, Giorgio (2016): Strategy-proof and fair assignment is wasteful. Games and Economic
Behavior 98:172--179.

\item Nesterov, Alexander S. (2017): Fairness and efficiency in strategy-proof object allocation
mechanisms. Journal of Economic Theory, 170:145--168.
\item P\'{a}pai, Szilvia (2000): Strategyproof assignment by hierarchical exchange. Econometrica 68:1403--1433.
\item Pycia, Marek and Peter Troyan (2026): The random priority mechanism is uniquely simple, efficient, and fair. Econometrica, forthcoming.
\item Pycia, Marek and Utku M. \"Unver (2017): Incentive compatible allocation and exchange of discrete resources. Theoretical Economics 12:
287--329.
\item Shapley, Lloyd S. and Herbert Scarf (1974): On cores and indivisibility.
Journal of Mathematical Economics 1:23--28.
\item Shende, Priyanka and Manish Purohit (2023): Strategy-proof and envy-free mechanisms for house allocation. Journal of Economic Theory 213: 105712.
%\item Zhang, Jun (2019): Efficient and fair assignment mechanisms are strongly group manipulable.
%Journal of Economic Theory 180:167--177.
\item Zhang, Jun (2023): On wastefulness of random assignments in discrete
allocation problems. Economic Theory 76:289--310.
%\item Zhang, Jun (2023b):
%Strategy-proof allocation with outside option. Games and Economic Behavior 137: 50--67.
\end{description}
%\bibliographystyle{plainnat}
%\bibliography{./library}

\end{document}